\documentclass{article}[11pt]

\usepackage{verbatim}
\usepackage{amsmath,amssymb,fullpage,xcolor,amsthm}
\usepackage[T1]{fontenc} \usepackage{fourier}

\usepackage{hyperref}

\newtheorem{lemma}{Lemma}[section]
\newtheorem{cor}{Corollary}
\newtheorem{prop}{Proposition}
\newtheorem{thm}{Theorem}
\newtheorem{defn}{Definition}
\newtheorem{rem}{Remark}

\newcommand{\vx}{\mathbf{x}}

\newcommand{\R}{\mathbb{R}}

\newcommand{\lpo}{\textsc{LP}^L}
\newcommand{\lpt}{\textsc{LP}^U}

\newcommand{\cutE}[1]{\delta\left({#1}\right)}
\newcommand{\abs}[1]{\left|{#1}\right|}

\newcommand{\cC}{\mathcal{C}}
\newcommand{\cS}{\mathcal{S}}
\newcommand{\ball}{\mathcal{B}}

\newcommand{\av}[2]{\mathbb{E}_{#2}\left[{#1}\right]}

\newcommand{\cost}{\textsc{cost}}

\newcommand{\ED}{\mathrm{ED}}
\newcommand{\xor}{\mathrm{XOR}}

\newcommand{\yell}[1]{}
\newcommand{\ayell}[1]{}

\newcommand{\cK}{\mathcal{K}}
\newcommand{\cM}{\mathcal{M}}
\newcommand{\cB}{\mathcal{B}}
\newcommand{\ST}{\mathrm{ST}}
\newcommand{\SD}{\mathrm{DISJ}}
\newcommand{\AND}{\mathrm{AND}}

\newcommand{\cD}{\mathcal{D}}

\newcommand{\eps}{\epsilon}

\newcommand{\maxhard}{maximally hard on the star graph}

\newcommand{\distrbnd}{\log{|V|}\log\log{|V|}}

\newcommand{\lpst}{LP_{\ST}}
\newcommand{\lpmfc}{LP_{\mathrm{MDN}}}
\newcommand{\lpmtch}{LP_{\mathrm{MTCH}}}

\newcommand{\IP}{\mathrm{IP}}
\newcommand{\TRIBES}{\mathrm{TRIBES}}
\newcommand{\UDISJ}{\mathrm{UDISJ}}

\newcommand{\binloss}{\log(nk)\cdot \log^2{k}}

\newcommand{\binlosslogk}{\log(nk)\cdot \log^3{k}}

\title{The Range of Topological Effects on Communication}

\author{Arkadev Chattopadhyay \thanks{School of Technology and Computer Science, Tata Institute of Fundamental Research, email: arkadev.c\@@tifr.res.in. Research partially supported by a Ramanujan Fellowship of the DST.} 
\and Atri Rudra\thanks{Department of Computer Science and Engineering, University at Buffalo, SUNY, email: atri\@@buffalo.edu. Research supported in part by NSF grant CCF-0844796.}}
\begin{document}
\maketitle


\thispagestyle{empty}

\begin{abstract}
We continue the study of communication cost of computing functions when inputs are distributed among $k$ processors, each of which is located at one vertex of a network/graph called a terminal. Every other node of the network also has a processor, with no input. The communication is point-to-point 
and the cost is the total number of bits exchanged by the protocol, in the worst case, on all edges. 

Chattopadhyay, Radhakrishnan and Rudra (FOCS'14) recently initiated a study of the effect of topology of the network on the total communication cost using tools from $L_1$ embeddings. Their techniques provided tight bounds for simple functions like Element-Distinctness (ED), which depend on the 1-median of the graph. This work addresses two other kinds of natural functions. We show that for a large class of natural functions like Set-Disjointness the communication cost is essentially $n$ times the cost of the optimal Steiner tree connecting the terminals. Further, we show for natural composed functions like  $\text{ED} \circ \text{XOR}$ and  $\text{XOR} \circ \text{ED}$, the naive protocols suggested by their definition is optimal for general networks. Interestingly, the bounds for these functions depend on more involved topological parameters that are a combination of Steiner tree and 1-median costs.  

To obtain our results, we use some new tools in addition to ones used in Chattopadhyay et. al. These include  (i) viewing the communication constraints via a linear program; (ii)  using tools from the theory of tree embeddings to prove  topology sensitive direct sum results that handle the case of composed functions and (iii) representing the communication constraints of certain problems as a family of collection of multiway cuts, where each multiway cut simulates the hardness of computing the function on the star topology. 
\end{abstract}

\newpage

\pagenumbering{arabic}

\section{Introduction} \label{sec:intro}

We consider the following distributed computation problem $p \equiv (f,G,K,\Sigma)$: there is a set $K$ of $k$ processors that have to jointly compute a function $f\,:\,\Sigma^{K} \to \{0,1\}$. Each of the $k$ inputs to $f$ is held by a distinct processor. \yell{We can also handle the case when a terminal has more than one input. Should we mention this as a footnote and concentrate on $K$ being a proper set and not a multi-set for the intro? --Atri} Each processor is located on some node of a network (graph) $G \equiv (V,E)$. These nodes in $V$ with an input are called {\em terminals} and the set of such nodes is denoted by $K$. The other nodes in $V$ have no input but have processors that also participate in the computation of $f$ via the following communication process: there is some fixed a-priori protocol according to which, in each round of communication, nodes of the network send messages to their neighbors. The behavior of a node in any round is just a (randomized) function of inputs held by it and the sequence of bits it has received from its neighbors in the past. All communication is point-to-point in the sense that each edge of $G$ is a private communication channel between its endpoints. In any round, if one of the endpoints of an edge is in a state where it expects to receive some communication from the other side, then silence from the other side is not allowed in a legal protocol. At the end of communication process, some pre-designated node of the network outputs the value of $f$ on the input instance held by processors in $K$. \yell{We need to clarify the silence equals one bits transmitted somewhere. Maybe in the prelims section? --Atri}\ayell{I have clarified it. While doing this I noticed that silence being counted as one bit and silence not being tolerated is a little different. In particular, silence tolerated but counted as 1 bit does not quite let us prove later the Steiner tree lower bound for all functions. This seems to make the protocol restrictive, but this is my quick way out for now:--Arkadev.} We assume that protocols are randomized, using public coins that are accessible to all nodes of the network, and err with small probability. The cost of a protocol on an input is the expected total number of bits communicated on all edges of the network. The main question we study in this work is how the cost of the best protocol on the worst input depends on the function $f$, the network $G$ and the set of terminals $K$. This cost is denoted by $R_{\epsilon}\big(p\big)$ (and we use $R(p)$ to denote $R_{1/3}(p)$). It is not difficult to see that this cost is lower bounded by the expected cost (of the best protocol) under any distribution $\mu$ over the inputs to nodes in $K$. This latter quantity is denoted by $R_{\epsilon,\mu}\big(p\big)$ and turns out to be easier to lower bound under a conveniently chosen $\mu$.

This communication model seems to be a natural abstraction of many distributed problems and was recently studied in its full generality by Chattopadhyay, Radhakrishnan and Rudra~\cite{CRR14}.\footnote{Related but different problems have been considered in distributed computing. Please see Appendix~\ref{app:dc} for more details.} A noteworthy special case is when $G$ is just a pair of nodes connected by an edge. This corresponds to the classical model of 2-party communication introduced by Yao~\cite{Yao79} more than three decades ago. The study of the classical model has blossomed into the vibrant and rich field of communication complexity, which has deep connections to theoretical computer science in general and computational complexity in particular.    

This point-to-point model had received early attention in the works of Tiwari~\cite{T87}, Dolev and Feder~\cite{DF89} and Duris and Rolim~\cite{DR98}. These early works seem to have entirely focused on deterministic and non-deterministic complexities. In particular, Tiwari~\cite{T87} showed several interesting topology-sensitive bounds on the cost of deterministic protocols for simple functions. However, these bounds were for specific graphs like trees, grids, rings etc. More recently, there has been a resurgence of interest in the randomized complexity of functions in the point-to-point model. These have several motivations: BSP model of Valiant~\cite{bgp}, models for MapReduce~\cite{mr-1}, parallel models to compute conjunctive queries~\cite{join-1}, distributed models for learning~\cite{learning-1}, distributed streaming and functional monitoring~\cite{func-survey}, sensor networks~\cite{KK12} etc. Interestingly, in a very recent work Drucker, Kuhn and Oshman \cite{DKO14} showed that some outstanding questions in this model (where one is interested in bounding the number of rounds of communication as opposed to bounding the total communication) have connections to well known hard problems on constant-depth circuits. \yell{Should we also mention other related work in distributed computing that are related but not exactly the problem we are considering? --Atri} \ayell{The DKO reference is a PODC paper, so takes care of that--Arkadev.} Motivated by such diverse applications, a flurry of recent works~\cite{PVZ12,WZ12,WZ13,WZ14,BEOPV13,HRVZ13,LSWW14,CM14} have proved strong lower bounds, developing very interesting techniques. All of these works, however, focus on the star topology with $k$ leaves, each having a terminal and a central non-terminal node. Note that every function on the star can be computed using $O(kn)$ bits of communication, by making the leaves simultaneously send each of their $n$-bit inputs to the center that outputs the answer. The aforementioned recent works show that this is an optimal protocol for various natural functions. 

In contrast, on a general graph not all functions seem to admit $O(kn)$-bit protocols. Consider the naive protocol that makes all terminals send their inputs to a special node $u$. The speciality of $u$ is the following: let the status of a node $v$ in network $G$ w.r.t. $K$, denoted by $\sigma_K\left(v\right)$, be given by $\sum_{w \in K} d_G(v,w)$, where $d_G(x,y)$ is the length of a shortest path in $G$ between nodes $x$ and $y$. Node $u$ is special and called the {\em median} as it has a minimal status among all nodes, which we denote by $\sigma_K\left(G\right)$. Thus, the cost of the naive protocol is $\sigma_K\left(G\right)\cdot n$. For the star, the center is the median with status $k$. On the other hand, for the line, ring and grid, each having $k$ nodes all of which are terminals, $\sigma_K(G)$ is $\Theta(k^2)$, $\Theta(k^2)$ and $\Theta(k^{3/2})$ respectively.

The work in \cite{CRR14} appears to be the first one to address the issue of randomized protocols over arbitrary $G$. It shows simple natural functions like Element-Distinctness\footnote{Given inputs $X^i\in \Sigma$ for every $i\in K$, the function $\ED:\Sigma^K\to\{0,1\}$ is defined as follows:
$\ED\left((X^i)_{i\in K}\right) =1 \text{ if and only if } X^i\neq X^j \text{ for every } i\neq j\in K$.},
 have $\Theta(\sigma_K\left(G\right))$ as the cost (up to a poly-log$(k)$ factor) of the optimal randomized protocol computing them. While these are essentially the strongest possible lower bounds\footnote{Strictly speaking, the strongest lower bound is $\Omega(\sigma_K\left(G\right)\cdot n)$. Several functions, called linear 1-median type later, are shown to achieve this bound in \cite{CRR14}.}, not all functions of interest have that high complexity. Consider the function Equality that outputs 1 precisely when all input strings at the nodes in $K$ are the same. There is a randomized protocol of cost much less than $\sigma_K(G)$ for computing it: consider a minimal cost Steiner-tree with nodes in $K$ as the terminals. Let the cost of this tree be denoted by $\ST\left(G,K\right)$. Root this tree at an arbitrary  node. Each leaf node sends a hash (it turns out $O(1)$ bits of random hash suffices for our purposes\footnote{Observe that if two strings held at two terminals are not equal, each hash will detect inequality with probability $2/3$.}) of its string to its parent. Each internal node $u$ collects all hashes that it receives from nodes in the sub-tree rooted at $u$, verifies if they are all equal to some string $s$. If so, it sends $s$ to its parent and otherwise, it sends a special symbol to its parent indicating inequality. Thus, in cost $O\left(\ST\left(G,K\right)\right)$, one can compute Equality with small error probability.\footnote{In fact, we observe in Theorem~\ref{thm:steiner-tree-minimal} that {\em any} function $f:\Sigma^K\to \{0,1\}$ that depends on all of its input symbols needs $\Omega(\ST(G,K))$ amounts of communication (even for randomized protocols), which implies that the randomized protocol above for Equality is essentially optimal.}

\yell{In most of the discussion below we are ignoring $n$, which is appropriate for the functions we consider. So maybe state this explicitly? --Atri}

For many scenarios in a distributed setting, the task to be performed is naturally layered in the following way. The set of terminal nodes is divided into $t$ groups $K_1,\ldots,K_t$. Within a group of $m$ terminals, the input needs to be pre-processed in a specified manner, expressed as a function $g: \big(\{0,1\}^n\big)^m \to \{0,1\}^n$. Finally the results of the computation of the groups need to be combined in a different way, given by another function $f : \big(\{0,1\}^n\big)^t \to \{0,1\}$. More precisely, we want to compute the composed function $f \circ g$. The canonical protocol will first compute in parallel all instances of the task $g$ in groups using the optimal protocol for $g$ and then use the optimal protocol for $f$ on the outputs of $g$ in each of $K_i$. However, this is not the optimal protocol for all $f,g$ and network $G$. For example, consider the case when $f$ is Equality and $g$ is the bit-wise XOR function. As we show later, the optimal protocol for computing XOR has cost $\Theta\left(\ST\left(G,K\right)\cdot n\right)$. Hence, the naive protocol for $\text{EQ}\circ \text{XOR}$ will have cost $\Omega\big(\left(\ST\left(G,K'\right)\right)\,+ \,\sum_{i=1}^r \left(\ST\left(G,K_i\right)\cdot n\right)\,\big)$. However, it is not hard to see that there is a protocol of cost $O\big(t\cdot \left(\ST\left(G,K\right)\right)$. This cost can be much lower than the naive cost depending on the network.  
\ayell{It will be good to give a counter-example. Isn't EQ composed with XOR a counter example?}

\section{Our Results}
%
The first part of our work attempts to understand when the naive protocol cannot be improved upon for composed functions. Function composition is a widely used technique in computational complexity for building new functions out of more primitive ones \cite{RM99,G11,GNW11,BBCR13,KRW95}. Proving that the naive way of solving $f \circ g$ is essentially optimal, in many models remain open. In particular, even in the 2-party model of communication where the network is just an edge, this problem still remains unsolved (see \cite{BBCR13}).
 To describe our results on composition, we need the following terminology: The cost of solving a problem $\big(f,G,K,\{0,1\}^n\big)$ will have a dependence on both $n$ and the topology of $G$. We will deal with two kinds of dependence on $n$. If the cost depends linearly on $n$, we say $f$ is of linear type. Otherwise, there is no dependence on $n$. (We typically ignore poly-log factors in this paper.)  Call $f$ a \emph{$1$-median} type function if its topology-sensitive complexity is $\sigma_K\left(G\right)$. We say $f$ is of \emph{Steiner tree} type, if its topology-sensitive complexity is $\ST\left(G,K\right)$. The protocol for a Steiner tree type problem $f$ seems to move information around in a fundamentally different way from the one for a $1$-median type problem $g$. It seems tempting to expect that there composition cannot be solved by any cheaper protocol than the naive ones. However, we are only able to prove this intuition for few natural instances in this work.    


Consider the following composition: the first function is element distinctness function, denoted by $\ED$, which was shown by \cite{CRR14} to be of $1$-median type. The second is the bit-wise xor  function (which we denote by $\xor_n$), which is shown to be of linear Steiner-tree type later in Appendix~\ref{app:lp-cc}. In particular, given a graph $G=(V,E)$ and $t$ subsets $K_1,\dots,K_t\subseteq V$, we define the composed function $\ED\circ\xor_n$ as follows.
Given $k_i\stackrel{def}{=} |K_i|$ $n$-bit vectors $X_1^i,\dots,X_{k_i}^i\in\{0,1\}^n$ for every $i\in [t]$, define
$\ED\circ\xor_n\left( X^1_1,\dots,X^1_{k_1},\dots,X^t_1,\dots,X_{k_t}^t\right)=\ED\left(\xor_n\left(X^1_1,\dots,X^1_{k_1}\right),\dots,\xor_n\left(X^t_1,\dots,X_{k_t}\right)\right)$.
The naive algorithm mentioned earlier specializes for $\ED\circ\xor_n$ as follows: compute the inner bit-wise $\xor$'s first\footnote{In fact, we just need to compute the $\xor$ of the hashes of the input, which with a linear hash is just the bit-wise $\xor$ of $O(\log{k})$-bits of hashes.} and then compute the $\ED$ on the intermediate values. This immediately leads to an upper bound of 
\begin{equation}
\label{eq:ed-xor-mixed-ub}
 O\left(\sigma_{K_1,\dots,K_t}(G)\cdot \log{k}+ \sum_{i=1}^t \ST(G,K_i)\cdot \log{k}\right),
\end{equation} where $\sigma_{K_1,\dots,K_t}(G)$ is the minimum of $\sigma_K(G)$ for every choice of $K$ that has exactly one terminal from $K_i$ for every $i\in [t]$. One of our results, stated below, shows that this upper bound is tight to within a poly-log factor: 

\begin{thm}
\label{thm:ed-xor-lb}
\[R(\ED\circ\xor_n,G,K,\{0,1\}^n)\ge \Omega\left(\frac{\sigma_{K_1,\dots,K_t}(G)}{\log{t}}\;+\; \frac{\sum_{i=1}^t \ST(G,K_i)}{\distrbnd}\right).\]
\end{thm}

We prove the above result (and other similar results) by essentially proving a topology sensitive direct sum theorem (see Section~\ref{sec:tree} for more).

To get a feel for how~\eqref{eq:ed-xor-mixed-ub} behaves between the two extremes consider the case when $G$ is a $\sqrt{k}\times \sqrt{k}$ grid and the set of $k$ terminals (i.e. all nodes are terminals) is divided into $t$ sets of size $k/t$, where each $K_i$ for $i\in [t]$ is a $\sqrt{\frac{k}{t}}\times\sqrt{\frac{k}{t}}$ sub-grid. It can be verified that in this case~\eqref{eq:ed-xor-mixed-ub} is (up to an $O(\log{k})$ factor) $t\sqrt{k}+k$. 
 In Section~\ref{sec:xor-ed}, we further show that changing the order of composition to $\text{XOR}\circ \text{ED}$ also does not allow any cost savings over the naive protocol:

\begin{thm}
\label{thm:xor-ed-lb}
For every choice of $u_i\in K_i$:
\[R(\xor_1\circ\ED,G,K,\{0,1\}^n)\ge \Omega\left(\ST(G,\{u_1,\dots,u_t\})+\frac{\sum_{i=1}^t \sigma_{K_i}(G)}{\log{k}}\right).\]
\end{thm}

The results discussed so far follow by appropriately reducing the problem on a general graph to a bunch of two-party lower bounds, one across each cut in the graph. This was the general idea in~\cite{CRR14} as well but the reductions in this paper need to use different tools. However, the idea of two-party reduction seems to fail for the Set-Disjointness function,
 which is one of the centrally studied function in communication complexity. In our setting, the natural definition of Set-Disjointness (denoted by $\SD$) is as follows: each of the $k$ terminals in $K$ have an $n$-bit string and the function tests if there is an index $i\in [n]$ such that all $k$ strings have their $i$th bit set to $1$. It is easy to check that this function can be computed with $O(\ST(G,K)\cdot n)$ bits of communication (in fact one can compute the bit-wise $\AND$ function with this much communication by successively computing the partial bit-wise $\AND$ as we go up the Steiner tree). Before our work, only a tight bound was known for the special case of $G$ being a $k$-star (i.e. a lower bound of $\Omega(kn)$), due to the recent work of Braverman et al.~\cite{BEOPV13}. In this work, we present a fairly general technique that ports a tight lower bound on a $k$-star to an almost tight lower bound for the general graph case. For the complexity of Set-Disjointness, this technique yields the following bound: 

\begin{thm}
\label{thm:sd-lb}
\[R(\SD,G,K,\{0,1\}^n)\ge \Omega\left(\frac{\ST(G,K)\cdot n}{\log^2{k}}\right).\]
\end{thm}

Next, we present our key technical results and an overview of their proofs. We would like to point out that our proofs use many tools used in algorithm-design like (sub)tree embeddings, Bor\.{u}vka's algorithm to compute an MST for a graph and integrality gaps of some well-known LPs, besides using $L_1$-embeddings of graph that was also used in \cite{CRR14}. We hope this work encourages further investigation of other algorithmic techniques to prove message-passing lower bounds.

\section{Key Technical Results and Our Techniques}
\label{sec:gen}

In Appendix~\ref{app:lp-cc} we present a simple formulation of communication lower bounds in terms of a linear program (LP), whose constraints correspond to two-party communication complexity lower bounds induced across various cuts in the graph $G$. 
In particular, we prove our earlier claimed lower bound of $\Omega(\ST(G,K)\cdot n)$ for the $\xor_n$ problem.
Further, this connection can also be used to recover the $\Omega(\sigma_K(G)/\log{k})$ lower bound for the $\ED$ function from~\cite{CRR14}-- see Theorem~\ref{thm:ed-lb}. While LPs have been used to prove communication complexity lower bounds in the standard 2-party setting (see e.g. \cite{Sherstov09,SZ09}), our use of LPs above seem to be novel for proving communication lower bounds. 
In the remainder of the section, we present two general results that we will use to prove our lower bounds for specific functions including those in Theorems~\ref{thm:ed-xor-lb},~\ref{thm:xor-ed-lb} and~\ref{thm:sd-lb}. (See Appendix~\ref{sec:app} for the details.)

\subsection{A Result on Two LPs}
\label{sec:tree}



We now present a result that relates the objective values of two similar LPs. Both the LPs will involve the same underlying topology graph $G=(V,E)$.

We begin with the first LP, which we dub $\lpo(G)$:
\[\min \sum_{e\in E} x_e\]
subject to
\begin{align*}
\sum_{e\text{ crosses } C} x_e  & \ge \sum_{i=1}^{\ell} b^i(C) & \text{ for every cut } C\\
x_e& \ge 0 &\text{ for every } e\in E.
\end{align*}
In our results, we will use $x_e$ to denote the expected communication of an arbitrary protocol for a problem $p$ over a distribution over the input. The constraint for each cut $C$ will correspond to a two-party lower bound of $\sum_{i=1}^{\ell} b^i(C)$. Then the objective value of the above LP, which by abuse of notation we will also denote by $\lpo(G)$, will be a valid lower bound on $R(p)$.

Next we consider the second LP, which we dub $\lpt(G)$:
\[\min \sum_{i=1}^{\ell} \sum_{e\in E} x_{i,e}\]
subject to
\begin{align*}
\sum_{e\text{ crosses } C} x_{i,e}  & \ge b^i(C) & \text{ for every cut } C\text{ and }i\in [\ell]\\
x_{i,e}& \ge 0 &\text{ for every } e\in E\text{ and } i\in [\ell].
\end{align*}
In our results, we will connect the objective value of the above LP (which again with abuse of notation we will denote by $\lpt(G)$) to  the total communication of a trivial algorithm that solves problem $p$.

Our main aim is to show that for certain settings, the lower bound we get from $\lpo(G)$ is essentially the same as the upper bound we get from $\lpt(G)$.

Before we state our main technical result, we need to define the property we need on the values $b^i(C)$. In particular, let $\delta(C)$ denote the set of crossing edges for a cut $C$. We say that the values $b^i(C)$ satisfy the {\em sub-additive property} if for any three cuts $C_1,C_2$ and $C_3$ such that $C_1\cup C_2= C_3$,\footnote{This means that one side of the cut $C_3$ is the union of one side each of $C_1$ and $C_2$.} we have that for every $i\in [\ell]$:
$b^i(C_3)\le b^i(C_1)+b^i(C_2)$.
We remark that the two main families of functions that we consider in this paper lead to LPs that do satisfy the sub-additive property (see Appendix~\ref{app:sub-add}).
We are now ready to state our first main technical result:
\begin{thm}
\label{cor:lp1-lp2}
For any graph $G=(V,E)$ (and values $b^i(C)$ for any $i\in [\ell]$ and cut $C$ with the sub-additive property), we have
\[\lpt(G)\ge \lpo(G)\ge \Omega\left(\frac{1}{\distrbnd}\right)\cdot \lpt(G).\]
\end{thm}

Theorem~\ref{cor:lp1-lp2} is the main ingredient in proving the lower bound for a $1$-median function composed with a Steiner tree function as given in Theorem~\ref{thm:ed-xor-lb} (see Appendix~\ref{sec:ed-xor}).
%
%
We can also use Theorem~\ref{cor:lp1-lp2} to prove nearly tight lower bound for composing a Steiner tree type function $\xor$ with a linear $1$-median function $\IP$ as well as another $1$-median function $\ED$. However, it turns out for these functions, we
can prove a better bound than Theorem~\ref{cor:lp1-lp2}. 
%
%
 %
In particular, using techniques developed in~\cite{CRR14}, we can prove lower bounds given in Theorem~\ref{thm:xor-ed-lb} and the one stated below (see Appendix~\ref{sec:xor-ip} 
 for details):
\begin{cor}
\label{cor:xor-ip-lb}
For every choice of $u_i\in K_i$:
\[R(\xor\circ\IP_n,G,K,\{0,1\}^n)\ge\Omega\left(\ST(G,\{u_1,\dots,u_t\})+\frac{\sum_{i=1}^t \sigma_{K_i}(G)\cdot n}{\log{k}}\right).\]
\end{cor}

\subsubsection{Proof Overview}
\label{sec:techniques-1}

We give an overview of our proof of Theorem~\ref{cor:lp1-lp2} (specialized to the proof of Theorem~\ref{thm:ed-xor-lb}).
While the LP based lower bound argument for $\xor_n$ in Appendix~\ref{app:lp-cc} is fairly straightforward things get more interesting when we consider $\ED\circ\xor_n$. It turns out that just embedding the hard distribution for $\ED$ from~\cite{CRR14}, one can prove a lower bound of just $\Omega\left(\frac{\sigma_{K_1,\dots,K_t}(G)}{\log{t}}\right)$ (see Lemma~\ref{lem:ed-xor-mfc-lb}). The more interesting part is proving a lower bound of $\tilde{\Omega}\left(\sum_{i=1}^t \ST(G,K_i)\right)$. It is not too hard to connect the {\em upper} bound of $\tilde{O}\left(\sum_{i=1}^t \ST(G,K_i)\right)$ to the following LP, which we dub $\lpst^U(G,K)$ (and is a specialization of $\lpt(G)$):
\[\min \sum_{i=1}^{t} \sum_{e\in E} x_{i,e}\]
subject to
\begin{align*}
\sum_{e\text{ crosses } C} x_{i,e}  & \ge 1 & \text{ for every cut } C\text{ that separates } K \text{ and }i\in [t]\\
x_{i,e}& \ge 0 &\text{ for every } e\in E\text{ and } i\in [t].
\end{align*}
Indeed the above LP is basically solving the sum of $t$ independent linear programs: call them $\lpst(G,K_i)$ for each $i \in [t]$. Hence, one can independently optimize each of these $\lpst(G,K_i)$ and then just put them together to get an optimal solution for $\lpst^U(G,K)$. This matches the claimed upper bounds since it is well-known that the objective value of  $\lpst(G,K_i)$ is $\Theta(\ST(G,K_i))$~\cite{vazirani}.

On the other hand, if one tries the approach we used to prove the lower bound for $\xor_n$, then one picks an appropriate hard distribution $\mu$ and shows that for every cut $C$ the induced two-party problem has a high enough lower bound. In this case, it turns out (see Section~\ref{sec:ed-xor}) that the corresponding two-party lower bound (ignoring constant factors) is the number of sets $K_i$ separated by the cut. Then proceeding as in the argument for $\xor_n$ if one sets $y_e$ to be the expected (under $\mu$) communication for any fixed protocol over any $e\in E$, then $(y_e)_{e\in E}$ is a feasible solution for the following LP, which we dub $\lpst^L(G,K)$ (and is a specialization of $\lpo(G)$):
\[\min \sum_{e\in E} x_e\]
subject to
\begin{align*}
\sum_{e\text{ crosses } C} x_e  & \ge v(C,K) & \text{ for every cut } C\\
x_e& \ge 0 &\text{ for every } e\in E,
\end{align*}
where $v(C,K)$ is the number of subsets $K_i$ that are separated by $C$. If we denote the objective value of the above LP by $\lpst^L(G,K)$, then we have an overall lower bound of $\Omega(\lpst^L(G,K))$. Thus, we would be done if we can show that $\lpst^L(G,K)$ and $\lpst^U(G,K)$ are close. It is fairly easy to see that $\lpst^L(G,K)\le \lpst^U(G,K)$. However, to prove a tight lower bound, we need an approximate inequality in the other direction. We show this is true by the following two step process:
\begin{enumerate}
\item First we observe that if $G$ is a tree $T$ then $\lpst^L(T,K)= \lpst^U(T,K)$.
\item Then we use results from embedding graphs into sub-trees to show that there exists a subtree $T$ of $G$ such that $\lpst^L(G,K)\approx \lpst^L(T,K)$ and $\lpst^U(G,K)\approx \lpst^U(T,K)$, which with the first step completes our proof. 
\end{enumerate}
We would like to remark on three things. First, our proof can handle more general constraints than those imposed by the Steiner tree LP. In particular, we generalize the argument above to prove Theorem~\ref{cor:lp1-lp2}.
 Second, to the best of our knowledge this result relating the objective values of these two similar LPs seems to be new. However, we would like to point out that our proof follows (with minor modifications) a similar structure that has been used to prove other algorithmic results via tree embeddings (e.g. in~\cite{AA97}). Third, we find it interesting to observe that the upper bound on the gap between the two LP's is the key step in accomplishing a distributed direct-sum like result.

\subsection{From Star to Steiner Trees}


We define a multicut $C$ of $K$ to be a collection of non-empty pair-wise disjoint subsets $C_1,\ldots,C_r$ of $K$. Each such subset is called an explicit set of $C$ and the (maybe empty) set $K \setminus \cup_{i=1}^r C_i$ is called its implicit set.
We will call $f:\Sigma^K\to \{0,1\}$ to be $h$-{\em \maxhard} if the following holds for any multicut $C$. There exists a distribution $\mu_C^f$ such that the expected cost (under $\mu_C^f$) of any protocol that correctly computes $f$ on the following star graph is $\Omega(|C|\cdot h(|\Sigma|))$: each leaf of the star has all terminals from an explicit set from $C$, no two leaves have terminals from the same explicit set and the center contains terminals from the implicit set.
The following is our second main technical result:
\begin{thm}
\label{thm:sd-st} Let $f$ be $h$-\maxhard$\text{.}$ Then
\[R(f,G,K,\Sigma)\ge \Omega\left(\frac{\ST(G,K)\cdot h(|\Sigma|)}{\log^2{k}}\right).\]
\end{thm}

The above result easily implies the lower bound (see Section~\ref{sec:disj} ) in Theorem~\ref{thm:sd-lb}. 
%
%
\yell{I think the parameters for $\TRIBES$ is wrong above: i.e. $K$ should probably be broken up further. --Atri}
\ayell{I removed the Tribes result for now because of what I mentioned in the email. --Arkadev}
Theorem~\ref{thm:sd-st} can also be used to prove a lower  bound similar to Theorem~\ref{thm:sd-lb} above for the Tribes function using the lower bound for Tribes on the star topology from~\cite{CM14}. 
 We defer the proof of this claim to the full version of the paper.

\subsubsection{Proof Overview}
\label{sec:techniques-2}

In all of the arguments so far, we reduce the lower bound problem on $(G,K)$ to a bunch of two party lower bounds induced by cuts. However, we are not aware of any hard distribution such that one can prove a tight lower bound that reduces the set disjointness problem to a bunch of two-party lower bounds. In fact, the only non-trivial lower bound for set disjointness, in the point-to-point model, that we are aware of is the $\Omega(kn)$ lower bound for the $k$-star by Braverman et al.~\cite{BEOPV13}. In particular, their proof does not seem to work by reducing the problem to two-party lower bounds. In this work, we are able to extend the set disjointness lower bound of~\cite{BEOPV13} to Theorem~\ref{thm:sd-lb}. 

We prove Theorem~\ref{thm:sd-lb} by modifying the argument in~\cite{CRR14} as follows. Essentially the idea in~\cite{CRR14} is to construct a collection of cuts such that essentially every edge participates in $O(\log{k})$ cuts and one can prove the appropriate two-party lower bound across each of the cuts in the collection so that when one sums up the contribution from each cut one gets the appropriate $\Omega(\sigma_K(G)/\log{k})$ overall lower bound. (These collection of cuts were obtained via Bourgain's $L_1$ embedding~\cite{Bou85,LLR95}. As mentioned earlier, this trick does not seem to work for set disjointness and it is very much geared towards $1$-median type functions). We modify this idea as follows: we construct a collection of {\em multi-cuts} such that (i) every edge in $G$ appears in at most one multi-cut and (ii) one can use lower bounds on star graph to compute lower bounds for the induced function on each multi-cut, which can then be added up.

The main challenge in the above is to construct an appropriate collection of multi-cuts that satisfy properties (i) and (ii) above. The main idea is natural: we start with balls of radius $0$ centered at each of the $k$ terminals and then one grows all the balls at the same rate. When two balls intersect, we combine the two balls and grows the larger ball appropriately. The multi-cut at any point of time is defined by the vertices in various balls. To argue the required properties, we observe that the algorithm above essentially simulates Bor\.{u}vka's algorithm~\cite{boruvka} on the {\em metric closure} of $K$ with respect to the shortest path distances in $G$. In other words, we show that the sum of the contributions of the lower bounds from each multi-cut is related to the MST on the metric closure of $K$ with respect to $G$, which is well-known to be closely related to $\ST(G,K)$ (see e.g.~\cite[Chap. 2]{vazirani}). It turns out that for set disjointness, one has to define $O(\log{k})$ different hard distributions (that depend on the structure of the multi-cuts above) and this is the reason why we lose a $O(\log{k})$ factor in our lower bound. (We lose another $O(\log{k})$ factor since we use lower bounds on the star topology.) To the best of our knowledge this is the first instance where the hard distribution actually depends on the graph structure-- most of our results as well as those preceding ours use hard distributions that are {\em independent} of the graph structure. This argument generalizes easily to prove Theorem~\ref{thm:sd-st}.

\section{Open Questions}
\label{sec:concl}

\yell{ARKADEV: I'm not sure how much stress we should put on open question. If we list too many of them would it make our current results look too preliminary? Am not sure. For now am listing the main open questions I can think of. --Atri}
\ayell{ATRI: I think the first two are fine. I am commenting out the third for now. It may have a simple counter-argument. --Arkadev}

We conclude by pointing out two of the many open questions that arise from our work:
\begin{enumerate}
\item Our two main technical tools are complementary. Theorem~\ref{cor:lp1-lp2} works for the case when the set of terminals  $K$ is divided into sets $K_1,\dots,K_t$ and one applies some inner functions on these $K_i$'s. Theorem~\ref{cor:lp1-lp2} allows us to prove a sort of direct sum result in this case. 
However, this technique reduces the problem on $(G,K)$ to a bunch of two-party lower bounds. On the other hand, Theorem~\ref{thm:sd-st} transforms the problem on $(G,K)$ to lower bounds on star graphs. However, this cannot prove a direct sum type lower bound (and also only handles Steiner tree type constraints). A natural question to ask is if one can get the best of both worlds, i.e. can we show a direct sum type lower bound of the kind $\Omega(\sum_{i=1}^t \ST(G,K_i))$ by reducing the problem to a bunch of lower bounds on the star topology?
\yell{Should we explicitly mention where the argument for Lemma~\ref{lem:multicut=>sd-lb} fails for multiple inner functions? --Atri}
\ayell{I think the above level of detail is appropriate. --Arkadev}
\item In this paper we only present results for specific $f\circ g$. 
 It would be nice to prove our conjecture from the introduction: if the inner function is a (linear) Steiner tree type and the outer function is a (linear) $1$-median type function, then the trivial two-stage algorithm is optimal for $f\circ g$. There are several avenues to pursue this. One such is to extend the {\em XOR lemma} (which corresponds to proving that the naive protocol is optimal for $\xor\circ g$) of Barak et al.~\cite{BBCR13} from the two-party communication setting to ours (as long as $g$ is of 1-median type).
\end{enumerate}

%
%

\bibliographystyle{plain}
\bibliography{two-lps,references-ar,references-ac}

\section*{Acknowledgments}

Thanks to Jaikumar Radhakrishnan for pointing out that cost of minimum Steiner tree can bound the communication complexity of a class of functions.
Many thanks to Anupam Gupta for answering our questions on tree embeddings and related discussions.

We would like to thank the organizers of the \href{http://www.dagstuhl.de/en/program/calendar/semhp/?semnr=14391}{2014 Dagstuhl seminar on Algebra in Computational Complexity} for inviting us to Dagstuhl, where some of the results in this paper were obtained.

AC is supported by a Ramanujan Fellowship of the DST and AR is supported in part by NSF grant CCF-0844796.

\section*{Notes on the Appendix}


Further, some of our results hold for the case when more than one input is assigned to the same terminal, i.e. we have a {\em multi-set} of terminals. In the appendix, we will use $\cK$ to denote the case of the set of terminals being a multi-set and $K$ to denote the case that the set of terminals is a proper set.

\appendix

\section{Related Work in Distributed Computing}
\label{app:dc}

Not surprisingly, the role of topology in computation has been studied extensively in distributed computing~\cite{peleg-book}. There are three main differences between works in this literature and ours. First, the main objective in distributed computation is to minimize the end to end delay of the computation, which in communication complexity terminology corresponds to the number of rounds need to compute a given function. By contrast, we mostly consider the related but different measure of the total amount of communication. Second,
  the effect of network topology on the cost of
  communication has been analyzed to quite an extent when the networks
  are \emph{dynamic} (see for
  example the recent survey of Kuhn and Oshman \cite{KO11}). By
  contrast, in this paper we are concerned with static networks of
  arbitrary topology. Finally, there has also been work on proving lower bounds for distributed computing on static networks, see e.g. the recent work of Das Sarma et al.~\cite{das-sarma}. This line of work differs from ours in at least two ways. First, their aim is to prove lower bounds on the number of rounds needed to compute, especially when the edges of the graph are capacitated. This paper, on the other hand, focuses on the total communication needed without placing any restriction on the capacities of the edges or the number of rounds involved.
Second, the kinds of functions considered in the distributed computing community (for recent papers see e.g.~\cite{das-sarma,L13,DKO14}) are generally of a different nature than the kinds of functions that we consider in this paper (which are more influenced by the functions typically considered in the communication complexity literature). For many functions in distributed computing,  the function $f$ itself depends on $G$ (e.g. computing the diameter of $G$, the cost of the MST of $G$ etc.) while all the functions we consider are independent of $G$-- indeed we want to keep the function $f$ the same and see how its communication complexity changes as we change $G$. Further, even for the case when $f$ is independent of $G$ (e.g. sorting) typically one has $k=n$  and $V=\cK$ while in our case we have arbitrary $\cK$ and $|V|$ and $n$ are independent parameters. (There is a very recent exception in~\cite{KNPR13}.)

\section{Communcation Complexity Lower Bounds via LPs}
\label{app:lp-cc}

A basic idea in our technique, is to understand the topological constraints placed on the communication demands of the problem by considering cuts of a graph. The general idea of using cuts for this purpose has appeared in many places before like network coding (ex: \cite{lili1,lili2,lili3} and function computation in sensor networks (ex: \cite{KK12}). But the idea of using several cuts rather than a single cut that we describe next is primarily borrowed from \cite{CRR14} (similar though slightly less general arguments were also made in~\cite{T87,PVZ12}).  The original problem $\big(f,\{0,1\}^n, G, K\big)$ naturally gives rise to a classical 2-party problem across a cut $C = (V^A,V^B)$, where $V^A,V^B$ partition the set of vertices $V(G)$. In the 2-party problem, Alice gets the inputs of the terminals in $K^A \equiv K \cap V^A$ and Bob gets the inputs of terminals in $K^B \equiv K \cap V^B$. Alice and Bob compute $f^C$, the induced problem on the cut. A protocol $\Pi$ solving $f$ induces protocol $\Pi'$ for Alice and Bob as follows: let $\delta(C)$ be the set of cut-edges. As long as $\Pi$ does not send any bits across any edge in $\delta(C)$, Alice and Bob simulate $\Pi$ internally with no communication to each other. If $\Pi$ communicates bits through edges in $\delta(C)$, Alice sends exactly those bits to Bob that were sent in $\Pi$ from vertices in $V^A$ to vertices in $V^B$ via some pre-determined encoding in $\Pi'$. Bob then sends to Alice the bits sent in the other direction in $\Pi$. Thus, the 2-party problem gets solved in essentially the same cost as the total number of bits sent over edges of $\delta(C)$ by $\Pi$. However, the simple thing to note is that if $f^C$ is known to have large 2-party communication complexity of $b(C)$, then that places a communication demand of $b(C)$ across the cut $C$. 

We would like to say that we understand the communication bottlenecks in the graph, as is often done in analyzing network flows, by specifying this demand $b(C)$ from our understanding of 2-party communication complexity. An obvious problem is the following: usually randomized 2-party communication complexity specifies ``worst-case'' complexity. The worst-case cost locally across each cut $C$ may not correspond to a globally consistent input. It was observed in \cite{CRR14} that there is a simple fix to this. We define a global input distribution $\mu$ such that the ``expected'' communication cost of the 2-party problem across cut $C$  w.r.t the induced distribution $\mu_C$ is $b(C)$. Then, the use of linearity of expectation helps us analyze the expected communication cost of the original problem. This idea was used in \cite{CRR14} by using a special family of cuts obtained from $L_1$ embeddings of graphs. This worked well to give 1-median type lower bounds, where the demand function $b(C)$ was of a specific type. In this work, we want to deal with more varied demand functions. It turns out to be more convenient and (in hindsight) more natural for us to write these two-party communication constraints as a linear program (LP). This helps us not only to recover the bounds for the 1-median type functions but also to obtain tight bounds for other types of functions.  

We illustrate the use of an LP in our setting by considering the bit-wise xor function: given inputs $X^i=(X^i_1,\dots,X^i_n)\in\{0,1\}^n$ for every $i\in K$, the function $\xor_n:\left(\{0,1\}^n\right)^K\to\{0,1\}^n$ is defined as follows:
$\xor_n\left((X^i)_{i\in K}\right) = \left( \left(\bigoplus_{i\in K} X^i_j\right)_{j=1}^n\right)$,
where $\oplus$ denote the boolean xor function. It is easy to see that we can compute this function by successively computing the bit-wise xor values of inputs along the Minimum Steiner tree for $K$, which implies an upper bound of $O(\ST(G,K)\cdot n)$.  
We now show how one can prove an $\Omega(\ST(G,K)\cdot n)$ lower bound for $\xor_n$. 
 Let $\mu$ be the hard distribution that assign an independent and uniformly random vector from $\{0,1\}^n$ to each of the $k$ terminals. Now fix any protocol $\Pi$ that correctly solves the $\xor_n$ function on $(G,K)$ on all inputs. Now consider a cut $C$ of $G$ that separates the terminal set $K$. If one now considers the induced two party problem, it is not too hard to see that if Alice gets the vectors on one side of $C$ and Bob gets the rest of the input then Alice and Bob are trying to solve the two-party bit-wise XOR function. In particular, Alice and Bob have two vectors\footnote{$A$ is the bit-wise xor of all the inputs on Alice's side and $B$ is the bit-wise XOR of all the input on Bob's side.} $A,B\in\{0,1\}^n$ and they want to compute $\xor_n(A,B)$. $\Pi$ thus induces a bounded error randomized protocol for Alice and Bob where they communicate only bits that $\Pi$ communicates on cut-edges. Further, the induced distribution $\mu_C$ on $(A,B)$ is the uniform distribution on $\{0,1\}^n\times \{0,1\}^n$.  It is not difficult to use an entropy argument and conclude that the two-party problem has an expected (under $\mu_C$) communication complexity lower bound of at least $\alpha\cdot n$ for some absolute constant $\alpha>0$. Now for every $e\in E$, define $x_e$ to be expected total communication through edge $e$ by $\Pi$ under $\mu$. Then the argument above and linearity of expectation implies that the expected total communication complexity of $\Pi$  (scaled down by a factor of $\alpha\cdot n$) is lower bounded by the objective value of the following linear program, which we will dub $\lpst(G,K)$:
\[\min \sum_{e\in E} x_e\]
subject to
\begin{align*}
\sum_{e\text{ crosses } C} x_e&\ge 1 &\text{ for every cut } C \text{ that separates } K\\
x_e&\ge 0&\text{ for every }e\in E.
\end{align*}
By abuse of notation let $\lpst(G,K)$ also denote the objective value of the above LP. It is well-known that $\lpst(G,K)$ is $\Theta(\ST(G,K))$ (see Theorem~\ref{thm:st-lp}), which with the above discussion implies the desired lower bound of $\Omega(\ST(G,K)\cdot \alpha\cdot n)=\Omega(\ST(G,K)\cdot n)$ for the $\xor_n$ problem. 


\section{More Details on Graph Parameters}
\label{sec:graphs}

It is well-known that $\ST(G,K)$ is closely related to  $\lpst(G,K)$ (see e.g.~\cite{vazirani}):
\begin{thm}
\label{thm:st-lp}
\[ \lpst(G,K) \le \ST(G,K)\le 2\cdot\lpst(G,K).\]
\end{thm}


The quantity $\sigma_k(G)$ is closely related to the following LP, which we will dub $\lpmfc(G,K)$:
\[\min \sum_{e\in E} x_e\]
subject to
\begin{align}
\label{eq:lpmfc-cons}
\sum_{e\text{ crosses } C} x_e&\ge \min(|C|,|K\setminus C|) &\text{ for every cut } C\\
x_e&\ge 0&\text{ for every }e\in E. \notag
\end{align}
By abuse of notation let $\lpmfc(G,K)$ also denote the objective value of the above LP.
The following result was implicitly argued in~\cite{CRR14}. For the sake of completeness we present a proof in Appendix~\ref{sec:graphs}.
\begin{thm}
\label{thm:mfc-lp}
\[ \lpmfc(G,K) \ge \Omega\left(\frac{\sigma_K(G)}{\log{k}}\right).\]
\end{thm}
\begin{proof}
Let $\cC$ be the collection of cuts in $G$ guaranteed by Bourgain's embedding~\cite{Bou85,LLR95} that has the following two guarantees: (i) Every edge is cut by $\beta=O(\log{k})$ cuts in $\cC$ and (ii) for every $u\neq v\in V$, the pair is separated by at least $d_G(u,v)$ cuts in $\cC$.

Using the constraint~\eqref{eq:lpmfc-cons} over all cuts in $\cC$, we get that for the optimal solution $(x_e)_{e\in E}$ for $\lpmfc(G,K)$ (we use property (ii) of $\cC$ in the third inequality)
\begin{align*}
\sum_{C\in\cC} \sum_{e\in \delta(C)} x_e &\ge \sum_{C\in\cC} \min(|C|,|K\setminus C|)\\
&\ge \sum_{C\in\cC} \frac{|C|\cdot|K\setminus C|}{k}\\
&= \sum_{C\in\cC} \frac{\left| \{ (u,v)|u\neq v\in K, C\text{ separates } (u,v)\}\right|}{k}\\
&= \frac{1}{k}\cdot \sum_{u\neq v\in K} \left|\{C\in\cC| C\text{ separates } (u,v)\}\right|\\
&\ge \frac{1}{k}\cdot \sum_{u\neq v\in K} d_G(u,v)\\
&= \frac{1}{k}\sum_{u\in K}\sum_{v\in K, v\neq u} d_G(u,v)\\
&= \frac{1}{k}\sum_{u\in K} \sigma_K(u)\\
&\ge \sigma_K(G).
\end{align*}
Finally by property (i) of $\cC$ we have that $\sum_{C\in \cC} \sum_{e\in \delta(C)} x_e \le \beta\cdot \sum_{e\in E} x_e$, which with the above inequality implies that $\sum_{e\in E} x_e\ge \sigma_K(G)/\beta$, as desired.
\end{proof}

\cite{CRR14} also considered another graph parameter. Given the graph $G=(V,E)$, the subset of even number of terminals $K$ and a partition $M$ of $K$ into sets of size exactly two, define $d(G,M)=\sum_{(u,v)\in M} d_G(u,v)$.
The quantity $d(G,M)$ is related to the following LP, which we will dub $\lpmtch(G,K,M)$:
\[\min \sum_{e\in E} x_e\]
subject to
\begin{align}
\label{eq:lpmtch-cons}
\sum_{e\text{ crosses } C} x_e&\ge m(C,M) &\text{ for every cut } C\\
x_e&\ge 0&\text{ for every }e\in E \notag,
\end{align}
where $m(C,M)$ is the number of pairs in $M$ separated by $C$.
By abuse of notation let $\lpmtch(G,K,M)$ also denote the objective value of the above LP.
The following result was implicitly argued in~\cite{CRR14}. For the sake of completeness we present a proof.
\begin{thm}
\label{thm:mtch-lp}
\[ \lpmtch(G,K,M) \ge \Omega\left(\frac{d(G,M)}{\log{k}}\right).\]
\end{thm}
\begin{proof}
Let $\cC$ be the collection of cuts in $G$ guaranteed by Bourgain's embedding~\cite{Bou85,LLR95} that has the following two guarantees: (i) Every edge is cut by $\beta=O(\log{k})$ cuts in $\cC$ and (ii) for every $u\neq v\in V$, the pair is separated by at least $d_G(u,v)$ cuts in $\cC$.

Using the constraint~\eqref{eq:lpmtch-cons} over all cuts in $\cC$, we get that for the optimal solution $(x_e)_{e\in E}$ for $\lpmtch(G,K,M)$ (we use property (ii) of $\cC$ in the last inequality)
\begin{align*}
\sum_{C\in\cC} \sum_{e\in \delta(C)} x_e &\ge \sum_{C\in\cC} \left| \{ (u,v)\in M| C\text{ separates } (u,v)\}\right|\\
&=  \sum_{(u,v)\in M} \left|\{C\in\cC| C\text{ separates } (u,v)\}\right|\\
&\ge \sum_{(u,v)\in M} d_G(u,v)\\
&= d(G,M).
\end{align*}
Finally by property (i) of $\cC$ we have that $\sum_{C\in \cC} \sum_{e\in \delta(C)} x_e \le \beta\cdot \sum_{e\in E} x_e$, which with the above inequality implies that $\sum_{e\in E} x_e\ge d(G,M)/\beta$, as desired.
\end{proof}

Finally, the quantities $\sigma_K(G)$ and the worst-case $d(G,M)$ are within a factor of $2$ of each other:
\begin{lemma}[\cite{CRR14}]
\label{lem:worst-match}
Let $K$ be a set of even number of terminals and let $\cM(K)$ denote the set of all disjoint pairings in $K$. Then
\[\frac{1}{2}\cdot \sigma_K(G)\le \max_{M\in\cM(K)}d(G,M)\le \sigma_K(G).\]
\end{lemma}

\section{Sub-additivity Property}
\label{app:sub-add}

We briefly argue that the two main families of functions that we consider in this paper lead to LPs that do satisfy the sub-additive property:
\begin{itemize}
\item Steiner Tree constraints. There are sets of terminals $T_i\subseteq V$ (for $i\in [\ell]$) and $b^i(C)=1$ if $C$ separates $T_i$ and $0$ otherwise. Note that with these constraints $\lpo(G)$ and $\lpt(G)$ are the same as $\lpst^L(G,K)$ and $\lpst^L(G,K)$ that we saw in the introduction.
\item Multi-commodity flow constraints. We have a set of demands $D_i$ (for $i\in [\ell]$)  and $b^i(C)$ is the number of demand pairs in $D_i$ that are separated by $C$. Note that with these constraints $\lpt(G)$ is essentially the sum of  $\lpmfc(G,K_i)$ where $D_i$ consists of all pairs in $K_i$.
\end{itemize}

\section{Applications}
\label{sec:app}

In this section, we apply the general techniques we have developed so far to obtain lower bounds for specific functions. However, we begin with our lower bound for all functions.

\subsection{Lower Bound For Every Function}

We prove here that every function needs $\Omega\left(\ST(K,G)\right)$ bits to be computed by any protocol.

\begin{thm}  \label{thm:steiner-tree-minimal}
Let $f:\Sigma^k \to \{0,1\}$ be any function that depends on all of its input symbols. Then, 
\[R\left(f,G,K,\Sigma\right) \geq \Omega(\ST\left(G,K)\right).\]
\end{thm}

\begin{proof}
Let $\Pi$ be any protocol in which the node $u$ in $V(G)$ computes the output of $f$. Take any cut $C$ of $G$ that partitions $V(G)$ into $V^A,\, V^B$ and separates the set $K$ of terminals into $K^A$ and $K^B$, each of which is non-empty. We argue that at least one bit is communicated in total across the edges of $\delta(C)$. This will be sufficient to establish our theorem, using Lemma~\ref{lem:lpo-cc-lb} and Theorem~\ref{thm:st-lp}.

WLOG, assume $u \in V^A$ is the designated terminal that needs to know the final output bit. As $f$ depends on all its input symbols, there is an assignment $a\in \Sigma^{K^A}$ to terminals in $K^A$ such that $f$ is determined by the assignment to terminals in $K^B$, i.e. there exists $b,b' \in \Sigma^{K^B}$ such that $f(a,b) \ne f(a,b')$. Hence, when $a$ is the assignment to $K^A$, there is at least 1 bit of communication across $\delta(C)$ from $V^B$ to $V^A$ for $u$ to output the answer correctly on all inputs with probability greater than 1/2. Otherwise, if no communication is expected by nodes in $V^A$, then the answer they give is independent of inputs to nodes in $K^B$. In this case, at least for one of the assignments $ab$ and $ab'$, protocol $\Pi$ errs with probability at least $1/2$.

For every other assignment to $K^A$, as long as there is no communication from $V^A$ to $V^B$, there is no way for processors in $V^B$ to know that the assignment to $K^A$ is not $a$. Hence, if they do not communicate in this case, they will also not communicate when $K^A$ is assigned $a$, which we argued is not possible. Thus, in every case, at least 1 bit of communication occurs on $\delta(C)$.
\end{proof}

We note that we only use the property of a valid protocol that it cannot have a deadlock (i.e. if one end point of an edge is expecting to receive some communication then the other end point has to communicate something) in the proof above. The rest of our proofs do not use this property explicitly.

Theorem~\ref{thm:steiner-tree-minimal} also has the following interesting consequence. Recall that in our model there is one designated terminal that needs to know the output bit. However, since the output bit can be transmitted to all the terminals with $\ST(G,K)$ amounts of additional communication, our model is equivalent (up to constant factors) to a related model where {\em all} the terminals need to know the output bit at the termination of the protocol.



\subsection{Bounds for $\SD$}
\label{sec:disj}

We next prove bounds for one of the most well-studied functions in classical communication complexity: the set disjointness function ($\SD$).

We first note that for a given set of terminals $K$, where each terminal gets a subset of $[n]$ (as a vector in $\{0,1\}^n$), one can compute their intersection by computing the running intersection from the leaves of the minimum Steiner tree on $K$ in $G$ to its root. Each edge only carries at most $n$ bits, which leads to the following result:

\begin{prop}
\label{prop:sd-ub}
\[R(\SD,G,K,\{0,1\}^n)\le O(\ST(G,K)\cdot n).\]
\end{prop}

Next, we argue that the bound above is nearly tight. Towards this end, we claim that
\begin{lemma}
\label{lem:sd-maxhard}
Let $h$ be the function $h(M)=\lceil \log{M}\rceil$. Then $\SD$ is $h$-\maxhard.
\end{lemma}
\begin{proof}
Let $\nu_s$ be the hard distribution for $\SD$ on an $s$-star from Braverman et al.~\cite{BEOPV13}. Using standard tools of information theory, it follows from that work that for such hard distribution the set disjointness problem needs $\Omega(sn)$ expected communication over an $s$-star. 

Now consider any multicut $C$ of $K$ with $|C|=s$. Now define $\mu_C^{\SD}$ as follows: let $(X_1,\dots,X_s)$ be a sample from $\mu_s$. Then the terminals in the $i$th explicit set in $C$ all get $X_i$. Finally, all the terminals in the implicit set get the all ones vector. It is easy to check that $\mu_C^{\SD}$ has the required property.
\end{proof}

The above lemma along with Theorem~\ref{thm:sd-st} immediately proves Theorem~\ref{thm:sd-lb}.


\subsection{Composed Functions}

Finally we prove bounds for some composed functions.

\subsubsection{Bounds for $\ED\circ\xor_n$}   \label{sec:ed-xor}

In this section, we consider the composed function $\ED\circ\xor_n$. For the sake of completenes, we recall its definition.
Let  $G=(V,E)$ be the graph and given $t$ subsets $K_1,\dots,K_t\subseteq V$ (which need not all be disjoint or distinct), we have  $\cK=\{K_1,\dots,K_t\}$.
Given $k_i\stackrel{def}{=} |K_i|$ $n$-bit vectors $X_1^i,\dots,X_{k_i}^i\in\{0,1\}^n$ for every $i\in [t]$, define:
\[\ED\circ\xor_n\left( X^1_1,\dots,X^1_{k_1},\dots,X^t_1,\dots,X_{k_t}^t\right)=\ED\left(\xor_n\left(X^1_1,\dots,X^1_{k_1}\right),\dots,\xor_n\left(X^t_1,\dots,X_{k_t}\right)\right).\]


We now state the obvious upper bound for solving the $\ED\circ\xor_n$ function. For notational convenience, define $\sigma_{K_1,\dots,K_t}(G)$ to be the minimum of $\sigma_K(G)$ for every choice of $K$ that has exactly one terminal from $K_i$ for every $i\in [t]$. Then we have the following upper bound.

\begin{prop}
\label{prop:ed-xor-ub}
Let $k=\sum_{i=1}^t k_i$. Then
\[R(\ED\circ\xor_n,G,\cK,\{0,1\}^n)\le O\left(\sigma_{K_1,\dots,K_t}(G)\cdot \log{k}+ \sum_{i=1}^t \ST(G,K_i)\cdot \log{k}\right).\]
\end{prop}
\begin{proof}
Note that with $O(\ST(G,K_i)\cdot \log{k})$ amounts of communication, every terminal in $K_i$ will know the hash\footnote{In particular, here the hash is the inner product of $O(\log{k})$ random vectors with the input. The random vectors are generated using public randomness.} of $\xor_n(X^i_1,\dots,X^i_{k_i})$. Doing this for every $i\in [t]$ gives the second term in the claimed bound.

Let $u_1,\dots,u_t$ be such that $u_i\in K_i$ for every $i\in [t]$ and $\sigma_{K_1,\dots,K_t}(G)=\sigma_{\{u_1,\dots,u_t\}}(G)$. Then run the upper bound protocol for $\ED$ using the hashes at  the terminals in the  set $\{u_1,\dots,u_t\}$. This latter part accounts for the first term in the claimed bound. This completes the proof.
\end{proof}


We will now prove Theorem~\ref{thm:ed-xor-lb}, which is an almost matching lower bound to the upper bound in Proposition~\ref{prop:ed-xor-ub}. We will do so by proving two lower bounds separately: one each for the two terms in the upper bound. Note that this immediately implies a lower bound that is the sum of the two terms (up to a factor of $1/2$) as desired.

The first term follows immediately from existing results~\cite{CRR14}:
\begin{lemma}
\label{lem:ed-xor-mfc-lb}
\[R(\ED\circ\xor_n,G,\cK,\{0,1\}^n)\ge \Omega\left(\frac{\sigma_{K_1,\dots,K_t}(G)}{\log{t}}\right).\]
\end{lemma}
\begin{proof}
Let $u_i\in K_i$ for every $i\in [t]$ be such that $\sigma_{K_1,\dots,K_t}(G)=\sigma_{\{u_1,\dots,u_t\}}(G)$. Let $\mu_t$ be the hard distribution from~\cite{CRR14} for $\ED$ on $t$ terminals. Assign the $t$ inputs from $\mu_t$ to each $u_i$ and all the other terminals in $\cup_{i=1}^t K_i$ get inputs that are distinct from each other and have a support disjoint from the support in $\mu_t$. Then the lower bound for $\mu_t$ from~\cite{CRR14} implies the claimed bound.
\end{proof}

\begin{rem}
\label{rem:max-sigma}
We note that the proof can also be extended to replace $\sigma_{K_1,\dots,K_t}(G)$ by the {\em maximum} $\sigma_{K'}(G)$, where $K'$ contains exactly one terminal from $K_1,\dots,K_t$. However, this does not lead to any contradiction since it is easy to check that $\sigma_{K'}(G)\le \sigma_{K_1,\dots,K_t}(G) +\sum_{i=1}^t \ST(G,K_i)$ and hence even if we use the stronger bound for above, the total lower bound does not exceed the upper bound.
\end{rem}

Next, we will prove a lower bound matching the second term in the upper bound in Proposition~\ref{prop:ed-xor-ub} up to poly-log factors. Before that we consider a specific problem that will be useful in the proof of our lower bound.

\begin{lemma}
\label{lem:ed-cor-2-party-lb}
Alice and Bob get $t$ inputs $A_1,\dots,A_t\in\{0,1\}^n$ and $B_1,\dots,B_t\in\{0,1\}^n$. They want to compute $\ED(\xor_n(A_1,B_1),\dots,\xor_n(A_t,B_t))$. Consider the distribution $\nu_t$ where each $A_i$ and $B_j$ are picked uniformly and independently at random. Then for $n\ge 3\log{t}$ and any protocol with bounded error that computes $\ED(\xor_n(A_1,B_1),\dots,\xor_n(A_t,B_t))$ correctly on all inputs has expected cost (under $\nu_t$) of $\Omega(t)$.
\end{lemma}
\begin{proof}
We will use the fact that the set disjointness problem where Alice and Bob get two sets of size $t$ where each set is picked by picking $t$ uniformly random elements (with replacement) from $\{0,1\}^n$ has expected communication complexity lower bound of $\Omega(t)$: see e.g.~\cite{CRR14}.\footnote{Technically in~\cite{CRR14} the hard distribution for set disjointness, the elements in the sets for Alice and Bob are chosen without replacement. However, the probability that either Alice or Bob have a set of size strictly less than $t$ or have an intersection is at most $\frac{\binom{2t}{2}}{2^n}$, which by our lower bound on $n$ is negligible.} Let us call his hard distribution $\mu_t$.

Now for the sake of contradiction assume that there exists a protocol $\Pi$ that computes $\ED(\xor_n(A_1,B_1),\dots,\xor_n(A_t,B_t))$ correctly on all inputs with expected cost (under $\nu_t$) $o(t)$. We will use this to obtain a protocol that solves the set disjointness problem above for sets of size $t/2$ with expected cost $o(t)$ under $\mu_{t/2}$, which will lead to a contradiction. Let us assume that Alice gets $\{X_1,\dots,X_{t/2}\}$ and Bob gets $\{Y_1,\dots,Y_{t/2}\}$ from the distribution $\mu_{t/2}$. Alice and Bob construct the sets $\{A_1,\dots,A_t\}$ and $\{B_1,\dots,B_t\}$ as follows. Using shared randomness Alice and Bob both pick uniformly random elements $Z_1,\dots,Z_t\in \{0,1\}^n$ and compute their sets as follows:
\[A_i=\left\{
\begin{array}{ll}
\xor_n(X_i,Z_i) & \text{if }i\le t/2\\
Z_i&\text{otherwise}
\end{array}
\right.
\]
and
\[B_i=\left\{
\begin{array}{ll}
Z_i & \text{if }i\le t/2\\
\xor_n(Y_{i-t/2},Z_i)&\text{otherwise}
\end{array}
\right.
.
\]
Note that the induced distribution on $\{A_1,\dots,A_t\}$ and $\{B_1,\dots,B_t\}$ is exactly the same as $\nu_t$. Further, we have $\ED(\xor_n(A_1,B_1),\dots,\xor_n(A_t,B_t))=1$ if and only if $\{X_1,\dots,X_{t/2}\}$ and $\{Y_1,\dots,Y_{t/2}\}$ are disjoint. Thus, if Alice and Bob run $\Pi$ on the inputs $A_1,\dots,A_t$ and $B_1,\dots,B_t$ as above, then they can solve the disjointness problem on inputs under the distribution $\mu_{t/2}$ with $o(t)$ expected cost, as desired.

\end{proof}

We are now ready to prove a matching lower bound for the second term in the upper bound in Proposition~\ref{prop:ed-xor-ub}.
\begin{lemma}
\label{lem:ed-xor-lb}
\[R(\ED\circ\xor_n,G,\cK,\{0,1\}^n)\ge\Omega\left(\frac{\sum_{i=1}^t \ST(G,K_i)}{\distrbnd}\right).\]
\end{lemma}
\begin{proof}
Consider the hard distribution $\mu$, where each of the inputs in $\cup_{i=1}^t K_i$ is chosen uniformly and independently at random from $\{0,1\}^n$. Now consider any cut $C$ in the graph $G$. Let $\ED\circ\xor_n(C)$ denote the induced two-party problem. We claim that this problems needs $\Omega(t')$ amounts of expected communication where $t'$ is the number of sets $K_i$ that are separated by $C$. Assuming this claim, note that by Corollary~\ref{cor:lpo-cc-lb} the effective lower bound for the entire problem is $\Omega(\lpo(G))$ where $\ell=t$ and $b^j(C)=1$ if $K_j$ is cut by $C$ and $0$ otherwise (for any $j\in [t]$). 
 Further, note that the values $b^j(C)$ are sub-additive. 
 By Theorem~\ref{cor:lp1-lp2}, we have a lower bound of
\[\Omega\left(\frac{\lpt(G)}{\distrbnd}\right).\]
To get the claimed lower bound, observe that the objective of $\lpt(G)$ is just the sum of $\lpst(G,K_i)$ for $i\in [t]$. 
Finally, since we can minimize the objective of $\lpt(G)$ by separately minimizing each instance of the Steiner tree LP, Theorem~\ref{thm:st-lp} implies that we have $\lpt(G)\ge \Omega\left(\sum_{i=1}^t \ST(G,K_i)\right)$, which implies the claimed lower bound.

We complete the proof by arguing the claimed lower bound on the two party function $\ED\circ\xor_n(C)$ for any cut $C$. WLOG assume that $C$ separates the sets $K_1,\dots,K_{t'}$. Then note that if Alice gets the inputs from one side of the cut $C$ and Bob gets the inputs from the other side then they are trying to solve $\ED(\xor_n(A_1,B_1),\dots,\xor_n(A_{t'},B_{t'}))$ where $A_i$ is the bit-wise $\xor_n$ of all inputs in $K_i$ that Alice gets and $B_i$ is the $\xor_n$ of the inputs from $K_i$ that Bob gets. Further, note that the distribution on $A_1,\dots,A_{t'}$ and $B_1,\dots,B_{t'}$ is the same as the hard distribution in Lemma~\ref{lem:ed-cor-2-party-lb}. Thus, Lemma~\ref{lem:ed-cor-2-party-lb} implies the claimed lower bound of $\Omega(t')$.
\end{proof}


\subsubsection{Bounds for $\xor\circ\IP$}
\label{sec:xor-ip}

There are multiple definitions of $\xor\circ\IP$ that make sense. In this subsection we will consider the version, which we dub $\xor\circ\IP_n$, that gives the cleanest bounds. Given the set of terminals $\cK$ divided into $t$ subsets of terminals $K_1,\dots,K_t$, let $M_i$ be a set of disjoint pairings of $K_i$ such that $d(G,M_i)=\Theta(\sigma_{K_i}(G))$ (by Lemma~\ref{lem:worst-match} such an $M_i$ exists). Given $k_i\stackrel{def}{=} |K_i|$ $n$-bit vectors $X_1^i,\dots,X_{k_i}^i\in\{0,1\}^n$ for every $i\in [t]$, define:
\[\xor\circ\IP_n\left( X^1_1,\dots,X^1_{k_1},\dots,X^t_1,\dots,X_{k_t}^t\right)=\xor_1\left(\IP_{M_1}\left(X^1_1,\dots,X^1_{k_1}\right),\dots,\IP_{M_t}\left(X^t_1,\dots,X_{k_t}\right)\right),\]
where $\xor_1$ denotes the function that first applied $\xor_n$ on the $t$ vectors and then takes the xor of the resulting $n$ bits and we consider the following version of the inner product function. Given a set of disjoint pairings $M$ of $K$, define $\IP_{M}:
\left(\{0,1\}^n\right)^K\to \{0,1\}^n$ as follows. Given inputs $X^i=(X^i_1,\dots,X^i_n)\in\{0,1\}^n$ for every $i\in K$,
$\IP_M\left((X^i)_{i\in K}\right) =  \left( \left(\bigoplus_{(u,v)\in M} \left(X^i_u\bigwedge X^i_v\right)\right)_{j=1}^n\right)$.

Now consider the obvious protocol to solve the $\xor\circ\IP_n$: first compute all the $\IP_{M_i}\left(X^i_1,\dots,X^i_{k_i}\right)$ using the trivial $\sigma_{K_i}(G)\cdot n$ protocol and then store the xor of the resulting $n$ bits at say $u_i\in K_i$. At this point with $O\left(\sum_{i=1}^t \sigma_{K_i}(G)\cdot n\right)$ bits of communication we have $t$ bits at $u_i\in K_i$. Then we compute the final desired output bits by using the Steiner tree on $\{u_1,\dots,u_t\}$. This implies an overall upper bound of

\begin{prop}
\label{prop:xor-ip-ub}
Let $k=\sum_{i=1}^t k_i$. Then
\[R(\xor\circ\IP_n,G,\cK,\{0,1\}^n)\le O\left(\ST(G,\{u_1,\dots,u_t\})+ \sum_{i=1}^t \sigma_{K_i}(G)\cdot n\right).\]
\end{prop}

Recall that Corollary~\ref{cor:xor-ip-lb} shows a nearly matching lower bound. One can easily show a matching lower bound for the first term in the sum above (e.g. by the argument for $\xor_n$ for $n=1$ from the introduction). We can also prove a nearly matching lower bound for the second term:
\begin{lemma}
\label{lem:xor-ip-lb}
\[R(\xor\circ\IP_n,G,\cK,\{0,1\}^n)\ge\Omega\left(\frac{\sum_{i=1}^t \sigma_{K_i}(G)\cdot n}{\log{k}}\right).\]
\end{lemma}

To prove this we will need the following result (the proof appears in Appendix~\ref{app:lem:lp1-lp2-mtch}):
\begin{lemma}
\label{lem:lp1-lp2-mtch}
Let $K_1,\dots,K_{\ell}$ and $M_1,\dots,M_{\ell}$ be defined as above. For any $j\in [\ell]$ define $b^j(C)$ is defined to be the number of pairs in $M_j$ separated by the cut $C$. Then for $k=\sum_{i=1}^{\ell} |K_i|$,
\[\Omega\left(\frac{\sum_{i=1}^{\ell} \sigma_{K_i}(G)}{\log{k}}\right)\le \lpo(G)\le \lpt(G)\le O\left(\sum_{i=1}^{\ell} \sigma_{K_i}(G)\right).\]
\end{lemma}

\begin{proof}[of Lemma~\ref{lem:xor-ip-lb}] Let $\mu$ be the distribution where the $k=\sum_{i=1}^t k_i$ vectors are picked uniformly and independently at random from $\{0,1\}^n$. Let $C$ be an arbitrary cut of $G$ and let $k'_i$ be the number of pairs in $M_i$ that are cut by $C$. Then note that the induced two-party problem is essentially trying to solve the two-party inner product function on $(\sum_{i=1}^t k'_i)\cdot n$ bits. Further, conditioned on all valid fixings of inputs corresponding to pairs that are not separated by $C$, the remaining inner product problem mentioned above corresponds to Alice (who receives all the vectors on one side of $C$) receiving a uniform vector with $(\sum_{i=1}^t k'_i)\cdot n$ uniform bits. Similarly for Bob. It is well-known \cite{CG88} that for this induced distribution the two party lower bound on the expected cost is $\Omega\left((\sum_{i=1}^t k'_i)\cdot n\right)$ bits of communication. 


Note that by Corollary~\ref{cor:lpo-cc-lb} the effective lower bound for the entire problem is $\Omega(\lpo(G))$ where $\ell=t$ and $b^j(C)= k'_j\cdot n$.
Lemma~\ref{lem:lp1-lp2-mtch} completes the proof.
\end{proof}

\subsubsection{Bounds for $\xor\circ \ED$}  \label{sec:xor-ed}





In this section, we consider in some sense the ``reverse'' of the $\ED\circ\xor_n$ function. The function $\xor_1\circ\ED:\left(\{0,1\}^n\right)^{\cK}\to \{0,1\}$ is defined as follows. Let the set of terminals $\cK$ be divided into $t$ subsets of terminals $K_1,\dots,K_t$. Given $k_i\stackrel{def}{=} |K_i|$ $n$-bit vectors $X_1^i,\dots,X_{k_i}^i\in\{0,1\}^n$ for every $i\in [t]$, define:
\[\xor_1\circ\ED\left( X^1_1,\dots,X^1_{k_1},\dots,X^t_1,\dots,X_{k_t}^t\right) =\bigoplus_{i=1}^t \ED\left(X^i_1,\dots,X^i_{k_i}\right).\]

Now consider the trivial two-step protocol that results in the following upper bound:
\begin{lemma}
\label{lem:xor-ed-ub}
Choose $t$ terminals $u_i\in K_i$ for every $i\in [t]$. Then
\[R(\xor_1\circ\ED,G,\cK,\{0,1\}^n)\le O\left(\ST(G,\{u_1,\dots,u_t\})+\sum_{i=1}^t \sigma_{K_i}(G)\cdot\log{k}\right).\]
\end{lemma}
\begin{proof}
Using the argument in proof of Proposition~\ref{prop:ed-xor-ub}, with $O\left(\sum_{i=1}^t \sigma_{K_i}(G)\cdot \log{k}\right)$ bits of communication, every $u_i$ knows the value of $ \ED\left(X^i_1,\dots,X^i_{k_i}\right)$. Then the resulting $\xor_1$ can be computed with $O(\ST(G,\{u_1,\dots,u_t\})$ bits of communication by progressively computing the $\xor_1$ along the corresponding Steiner tree.
\end{proof}

Recall that Theorem~\ref{thm:xor-ed-lb} shows a nearly matching lower bound.
We can have matching lower bound term for the first term in the sum above from Theorem~\ref{thm:steiner-tree-minimal}.\footnote{Technically, we get a lower bound of $\Omega(\ST(G,\cK))$, which of course implies a lower bound of $\Omega(\ST(G,\{u_1,\dots,u_t\}))$.} The more interesting part is to prove matching lower bound for the second term. Towards that end, we will need a result on classical 2-party Set-Disjointness: let $\UDISJ_n$ the unique-set disjointness problem on $2n$ bits that has the following promise. Alice and Bob get $n$-bit strings such that they have at most one occurrence of an all-one column in their inputs, i.e. their sets have at most one element in common. They want to find out if their sets intersect. Pair the input bits of Alice and Bob as $(X_1,Y_1),\ldots,(X_n,Y_n)$. Each pair $(X_i,Y_i)$ is sampled independently from a distribution $\mu$ that we describe next. To draw a sample $(U,V)$ from $\mu$, we first throw a uniformly random coin $D$. If $D=0$, $U$ is fixed to 0 and $V$ is drawn uniformly at random from $\{0,1\}$. If $D=1$, the roles of $U$ and $V$ are reversed. The following result was observed by \cite{CRR14}, using the seminal work of Bar-Yossef et al.~\cite{BJKS}:

\begin{thm}  \label{thm:info-set-disj}
Let $\Pi$ be any 2-party randomized protocol solving $\UDISJ_n$ with bounded error $\epsilon$ < 1/2. Then, its expected communication cost w.r.t. input distribution $\mu^n$ is at least $\big(1-2\sqrt{\epsilon}\big)(n/4)$.
\end{thm}

We are now ready to prove a nearly tight lower bound for the second term in Lemma~\ref{lem:xor-ed-ub}:

\begin{lemma}
\label{lem:xor-ed-lb} For $n\ge \log{k} + 2$,
\[R(\xor_1\circ\ED,G,\cK,\{0,1\}^n)\ge \Omega\left(\frac{\sum_{i=1}^t \sigma_{K_i}(G)}{\log{k}}\right).\]
\end{lemma}

\begin{proof}
We assume for convenience that each $|K_i|$ is even. Consider pairing $M_i$ of nodes in $K_i$, for each $i$ such that $d(G,M_i) \ge (1/2)\cdot \sigma_{K_i}\big(G\big)$ (such an $M_i$ exists thanks to Lemma~\ref{lem:worst-match}). Let $M$ be the multi-set union  $\cup_{i=1}^{\ell}M_i$, with $|M|= k/2 \equiv m$. Now for ease of description, we notate the inputs at the pairs of terminals in $M$ as $(X_1,Y_1),\ldots,(X_m,Y_m)$. We fix the first $\log k$ bits of each of the pair of terminals $X_j,Y_j$ to a string $a_j \in \{0,1\}^{\log k}$ such that $a_j \ne a_i$ for $i\ne j$. We call $a_j$ the prefix string of its pair. In the ensuing discussion we look at only restricted inputs, where the first $\log k$ bits of the inputs of each terminal are fixed to its respective prefix string. To keep notation simple, we still notate the unfixed bits of the $i$th  pair in $M$ as $(X_i,Y_i)$.  

We now describe the remaining input distribution: Let $\{s_x^0,s_y^0,s^1\}$ be three distinct strings in $\{0,1\}^{n'}$ where $n'= n-\log{k}$. Such three strings exist because of our assumed bound on $n$. Define auxiliary random variables $D_1,\ldots,D_m$ that are i.i.d and each takes value in $\{0,1\}$ uniformly at random. Then if $D_i = 0$, set $X_i=s_x^0$ and $Y_i$ takes uniformly at random a value in $\{s_y^0,\,s^1\}$. If $D_i=1$, then $Y_i = s_y^0$ and $X_i$ at random takes value in $\{s_x^0,\,s^1\}$. This completes the description of our input distribution that we denote by $\cD$.

We will show that we can invoke Lemma~\ref{lem:lp1-lp2-mtch} using distribution $\cD$. To do so, we analyze the expected communication cost of any protocol $\Pi$ solving $\xor_1 \circ \ED$ on $G$, across a cut $C$. Let number of pairs of $M_i$ cut by $C$ be $m'_i$ and $m' \equiv  m'_1 + \dots + m'_{\ell}$. Let $\cK_C$ denote the set of terminals whose mate in $M$ is separated by $C$. Let $\overline{\cK_C} \equiv \cK \setminus \cK_C$. Consider any assignment $\alpha$ to the terminals in $\overline{\cK_C}$ that is supported by $\cD$ and let the induced protocol be denoted by $\Pi_{\alpha}$.  We claim that we can solve unique (2-party) set-disjointness over $m'$ bits using $\Pi_{\alpha}$ as follows: Alice and Bob associate each of their co-ordinates with a separated pair in $M_C$.  Alice and Bob both replace their 1's by the string $s^1$. Alice replaces her 0's by $s_x^0$ and Bob replaces his by $s_y^0$. Then they simulate $\Pi_{\alpha}$ and communicate to each other whenever and whatever $\Pi_{\alpha}$ communicates across $C$. It is simple to verify that this way Alice and Bob can solve unique Set-Disjointness: if there is no all-1 column in their input, $\Pi_{\alpha}$ outputs  $r\mod{2}$ w.h.p, where $r$ is the number of $M_i$'s that are separated by $C$. If they do have a (unique) all-1 column,  $\Pi_{\alpha}$ outputs $(r-1)\mod{2}$ w.h.p. Further, the distribution induced on inputs of terminals in $\cK_C$, when Alice and Bob's input distribution is sampled from $\mu^{m'}$ (recall definition of $\mu$ from Theorem~\ref{thm:info-set-disj}), is precisely the distribution $\cD$ induces on $\cK_C$, conditioned on $\alpha$ assigned to inputs in $\overline{\cK_C}$. 

Thus, by Theorem~\ref{thm:info-set-disj}, the expected communication of $\Pi_{\alpha}$ over the cut edges of $C$ is $\Omega(m')$, for any $\alpha$. Hence, expected communication of $\Pi$ over $C$ is $\Omega(m')$. 
Corollary~~\ref{cor:lpo-cc-lb} and Lemma~\ref{lem:lp1-lp2-mtch} complete the proof.

\end{proof}

\section{Omitted Proofs from Section~\ref{sec:tree}}

\subsection{Proof of Theorem~\ref{cor:lp1-lp2}}
 
We first state a simple property of sub-additive values:
\begin{lemma}
\label{lem:sub-add-sum}
Let $G=(V,E)$ be a tree and let $b^i(C)$ for $i\in [\ell]$ be the constraint values for $\lpt(G)$ that satisfy the sub-additive property. For any edge $e\in E$, let $C_e$ denote the cut formed by removing $e$ from $G$. Then for any cut of $G$ we have
\[\sum_{e\in \delta(C)} b^i(C_e) \ge b^i(C).\]
\end{lemma}
\begin{proof}
This follows from the fact that $C=\cup_{e\in \delta(C)} C_e$ (since $G$ is a tree)  and the definition of the sub-additive property.
\end{proof}

In the rest of this subsection, we will prove Theorem~\ref{cor:lp1-lp2}.
We begin with the upper bound in Theorem~\ref{cor:lp1-lp2}, which is trivial. 
\begin{lemma}
\label{lem:trivial-ineq-G}
For any graph $G$,
\[\lpo(G)\le \lpt(G).\]
\end{lemma}
\begin{proof}
Consider any feasible solution $\{\vx_i\}_{i=1}^{\ell}$ for $\lpt(G)$, where $\vx_i=(x_{i,e})_{e\in E}$. Then note that the vector $\vx=(x_e)_{e\in E}$ defined as
\[x_e=\sum_{i=1}^{\ell} x_{i,e}\]
is also a feasible solution for $\lpo(G)$.
\end{proof}

To complete the proof, we now focus on proving the lower bound in Theorem~\ref{cor:lp1-lp2}.
We first begin by observing that the two LPs are essentially the same when $G$ is a tree:
\begin{lemma}
\label{lem:eq-tree}
For any tree $T=(V,E)$ (and values $b^i(C)$ for any $i\in [\ell]$ and cut $C$ with the sub-additive property), we have
\[\lpo(T)=\lpt(T).\]
\end{lemma}
\begin{proof}
The proof basically follows by noting that for a tree $T$, we only need to consider some special cuts. In particular, for every edge $e\in E$, let $C_e$ denote the cut that only cuts the edge $e$ (In other words, the two sides of the cut are formed by the two subgraphs obtained by removing $e$ from $T$).

We first claim that
\[\lpo(T)\ge \sum_{e\in E} \sum_{i=1}^{\ell} b^i(C_e).\] 
To see this consider any feasible solution $\vx\in \R^{E}$ for $\lpo(T)$. We have from the constraint on $C_e$ for every $e\in E$ that
\[x_e\ge \sum_{i=1}^{\ell} b^i(C_e).\]
Summing the above over all $e\in E$ completes the claim.

Finally, we argue that
\[ \lpt(T)\le \sum_{e\in E}\sum_{i=1}^{\ell} b^i(C_e),\]
which with Lemma~\ref{lem:trivial-ineq-G} will complete the proof. Consider the specific vector $\{\vx_i\}_{i\in [\ell]}$ such that for every $i\in [\ell]$ and $e\in E$, we have
\[x_{i,e}=b^i(C_e).\]
Note that the proof will be complete if we can show that the above vector is a feasible solution for $\lpt(T)$. Notice that by the fact that $T$ is a tree the above vector indeed does satisfy all the constraints corresponding to the cuts $C_e$ for every $e\in E$. Now consider an arbitrary cut $C$. Indeed we have for every $i\in [\ell]$:
\[\sum_{e\in \delta(C)} x_{i,e} = \sum_{e\in \delta(C)} b^i(C_e) \ge b^i(C),\]
where the inequality follows from Lemma~\ref{lem:sub-add-sum}. 
\end{proof}


Thus, we are done for the case when $G$ is a tree. For the more general case of a connected graph $G$, we will just embed $G$ into one of its sub-tree with a low distortion. This basically follows a similar trick used in~\cite{AA97}. We say a graph $G$ embeds with a distortion $\alpha$ on to (a distribution $\cD$ on) its subtrees such that for every $(u,v)\in V$, we have $\alpha\cdot d_G(u,v)\ge \av{d_T(u,v)}{T\gets \cD}$. (Note that for every sub-tree $T$ of $G$, we have $d_T(u,v)\ge d_G(u,v)$.)

We will now prove the following result:
\begin{lemma}
\label{lem:embedding}
Let $G=(V,E)$ embed into its subtrees under distribution $\cD$ with distortion $\alpha$. Then we have
\[\lpo(G)\ge \frac{1}{\alpha} \cdot \lpt(G).\]
\end{lemma}
\begin{proof}
Using the embedding trick of~\cite{AA97}, we will show that there exists a subtree $T$ of $G$ such that
\[\lpo(T)\le \alpha\cdot \lpo(G) \text{ and } \lpt(G)\le \lpt(T).\]
Note that the above along with Lemma~\ref{lem:eq-tree} completes the proof. Further, note that the second inequality in the above just follows from the fact that $T$ is a sub-tree of $G$. Hence, to complete the proof we only need to prove the first inequality.

A word of clarification. When we talk about the constraints in $\lpo(T)$ and $\lpo(G)$, we have the same $b^i(C)$ value for each cut. However, note that the set $\delta(C)$ could be different for $G$ and $T$.

Towards this end, consider an optimal solution $\vx\in\R^E$ for $\lpo(G)$. From this, we will construct a feasible solution $\vx'\in\R^E$ for $\lpo(T)$ whose expected cost is bounded, i.e.
\begin{equation}
\label{eq:embed-ineq}
\av{\sum_{e\in E(T)} x'_e}{T\gets \cD}\le \alpha\cdot \sum_{e\in E(G)} x_e.
\end{equation}
Markov's inequality will then complete the proof.

Finally, we define the solution $\vx'$ for $\lpo(T)$. Consider the following algorithm (for any given $T$):
\begin{enumerate}
\item Initialize $x'_e\gets 0$ for every $e\in E$.
\item For every $e=(u,v)\in E$ such that $x_e>0$ do the following
\begin{enumerate}
\item For the unique path $P_{u,v}$ that connects $u$ and $v$ in $T$ do the following
\begin{itemize}
\item For every $e'\in P_{u,v}$, do $x'_{e'}\gets x'_{e'}+x_e$.
\end{itemize}
\end{enumerate}
\end{enumerate}

We first argue that the vector $\vx'$ computed by the algorithm above is a feasible solution to $\lpo(T)$. Consider an arbitrary cut $C$ in $G$ and consider any $e=(u,v)\in \delta(C)$ such that $x_e>0$. 
Now consider the same cut $C$ in $T$. Note that in this case there has to be at least one edge $e'\in P_{u,v}$ such that $e'\in\delta(C)$ in $T$. Thus, we have
\[\sum_{e'\in\delta_T (C)} x'_e\ge \sum_{e\in \delta_G(C):x_e>0} x_e\ge \sum_{i=1}^{\ell} b^i(C),\]
where the last inequality follows since $\vx$ is a feasible solution for $\lpo(G)$. Thus, we have shown that $\vx'$ is a feasible solution.

Finally, we prove \eqref{eq:embed-ineq}. Note that by the algorithm above, we have 
\[\sum_{e\in E(T)} x'_e = \sum_{e=(u,v)\in E(G)} d_T(u,v)\cdot x_e.\]
Now \eqref{eq:embed-ineq} follows from the above, linearity of expectation and the fact that $G$ embeds with a distortion of $\alpha$ under $\cD$.
\end{proof}

It is known that any graph $G=(V,E)$ can be embedded into a distribution of its subtrees with distortion $O\left(\distrbnd\right)$ (see e.g.\cite{AN12,ABN08}), which in turn proves Theorem~\ref{cor:lp1-lp2}.\footnote{The result in~\cite{AN12} is not stated as distribution over sub-trees but rather the paper presents a deterministic algorithm to compute a tree $T$ that has low weighted average stretch. In our application, this means that the algorithm can compute a tree $T$ such that given weight $x_e$ for $e\in E(G)$, it is true that $\sum_{e=(u,v)\in E(G)} d_T(u,v) x_e\le \alpha \cdot\sum_{e\in E(G)} x_e$, which is enough for the rest of our proof to go through.}

\begin{rem}
\label{rem:just-trees}
It is natural to wonder if one can use embedding of a graph into a distribution of trees (instead of sub-trees as we do) and not lose the extra $\log\log{|V|}$ factor (since for trees one can get a distortion of $O(\log{|V|})$~\cite{FRT04}). We do not see how to use this result: in particular, in our proof of Lemma~\ref{lem:embedding} we do not see how to guarantee that the vector $(x'_e)_{e\in E(T)}$ satisfies the corresponding $\lpo(T)$ constraint. In short, this is because the edges in $T$ for the result in~\cite{FRT04} have weights  (say $w_e$ for every $e\in E(T)$) so we can no longer prove the (stronger) inequality $\sum_{e'\in\delta_T (C)} w_e\cdot x'_e\ge \sum_{e\in \delta_G(C):x_e>0} x_e$.
\end{rem}

\yell{ARKADEV: two questions. One: Am I missing something in the above? Two: should we keep the above remark? I think the question is natural enough and it might be worth it to explain why just tree embeddings do not work but am not sure how useful the above is. --Atri}
\ayell{At a quick glance, it seems you are right. I think it is worth keeping it: --Arkadev}

\subsection{Proof of Lemma~\ref{lem:lp1-lp2-mtch}}
\label{app:lem:lp1-lp2-mtch}

\begin{proof}[Proof of Lemma~\ref{lem:lp1-lp2-mtch}]
The inequality $\lpo(G)\le \lpt(G)$ follows from Lemma~\ref{lem:trivial-ineq-G}.

We begin with the last inequality. Towards this end we present a feasible solution for $\lpt(G)$. Fix an $i\in [\ell]$. Now consider the following algorithm to compute $x_{i,e}$ for $e\in E$:
\begin{itemize}
\item $x_{i,e}\gets 0$ for every $e\in E$.
\item For every $(u,v)\in M_i$, let $P_{u,v}$ be a shortest path from $u$ to $v$ in $G$. For every $e\in P_{u,v}$, do $x_{i,e}\gets x_{i,e}+1$.
\end{itemize}
It is easy to check that the vector computed above satisfies $\sum_{e\in E} x_{i,e}=d(G,M_i)\le O(\sigma_{K_i}(G))$ (where the inequality follows from our choice of $M_i$). Now consider any cut $C$. For every pair $(u,v)\in M_i$ that is cut by $C$, the chosen path $P_{u,v}$ will cross $C$ at least once. This implies that the vector $(x_{i,e})$ satisfies all the relevant constraints. This implies the claimed upper bound of $\lpt(G)\le O\left(\sum_{i=1}^{\ell} \sigma_{K_i}(G)\right)$.

We finally, argue the first inequality. We first note that $\lpo(G)$ is exactly the same as $\lpmtch(G,\cK,M)$ (where $M$ is the (multi-set) union of $M_1,\dots,M_{\ell}$). Thus, we have
\[\lpo(G)=\lpmtch(G,\cK,M) \ge \Omega\left(\frac{d(G,M)}{\log{k}}\right) = \Omega\left(\frac{\sum_{i=1}^{\ell} d(G,M_i)}{\log{k}}\right) \ge \Omega\left(\frac{\sum_{i=1}^{\ell} \sigma_{K_i}(G)}{\log{k}}\right),\]
where the first inequality follows from  Lemma~\ref{thm:mtch-lp} (and noting that its proof also works for the case when $M$ is a multi-set), the second equality follows from the fact that $M$ is the multi-set union of $M_1,\dots,M_t$ and the last inequality follows from our choices of $M_i$. This completes the proof.
\end{proof}

\subsection{Relating Communication Complexity Lower Bounds to $\lpo(G)$}

We now make the straightforward connection between two party lower bounds and $\lpo(G)$. In what follows consider a problem $p=(f,G,\cK,\Sigma)$. Further for any cut $C$ in graph $G$, we will denote by $f_C$ the two-party problem induced by the cut: i.e. Alice gets all the inputs from terminals in $\cK$ that are on one side of the cut and Bob gets the rest of the inputs. Finally, for a distribution $\mu$ over $\Sigma^{\cK}$ let $\mu_C$ be the induced distribution on the inputs on two sides of the cut.

\begin{lemma}
\label{lem:lpo-cc-lb}
Let $p=(f,G,\cK,\Sigma)$ be a problem and $\mu$ be a distribution on $\Sigma^{\cK}$ such that the following holds for every cut $C$ in $G$
\begin{equation}
\label{eq:2-party-lb}
D_{1/3,\mu_C}(f_C)\ge \sum_{i=1}^{\ell} b^i(C),\footnote{For a two party function $f$ and a distribution $\mu$ on the inputs of $f$, we will use $D_{\eps,\mu}(f)$ to denote the minimum expected communication cost over the distribution $\mu$ for the worst-case inputs over all protocols that compute $f$ with probability at least $1-\eps$ on every input.}
\end{equation}
then the following lower bound holds
\[R(p)\ge \lpo(G).\]
\end{lemma}
\begin{proof}
Let $\Pi$ be an arbitrary protocol that correctly solves the problem $p=(f,G,\cK,\Sigma)$ with error at most $\eps=1/3$. For a given input $Y\in  \Sigma^{\cK}$, let $c_e(Y,\Pi)$ denote the total amount of bits communicated over the edge $e\in E(G)$ for the input $Y$. For every $e\in E(G)$, define
\[x_e=\av{c_e(Y,\Pi)}{Y\gets \mu}.\]
Note that by linearity of expectation $\sum_{e\in E(G)} x_e$ denotes the expected cost of $\Pi$ on $p$. Further, if we can show that the vector $\vx=(x_e)_{e\in E(G)}$ as defined above is a feasible solution for $\lpo(G)$, then the expected communication cost of $\Pi$ will be lower bounded by $\lpo(G)$. The claim then follows since we chose $\Pi$ arbitrarily.

To complete the proof, we need to show that $\vx$ satisfies all the constraints. It follows from definition that $x_e\ge 0$ for every $e\in E(G)$. Thus, to complete the proof we need to show that for every cut $C$
\begin{equation}
\label{eq:constraint-sat}
\sum_{e\in \delta(C)} x_e\ge \sum_{i=1}^{\ell} b^i(C).
\end{equation}
Towards this end fix an arbitrary cut $C$ and consider the following protocol $\Pi_C$ for the induced two-party function $f_C$. Alice runs $\Pi$ by herself as long as $\Pi$ only uses messages on edges on Alice's side of the cut $C$. If $\Pi$ needs to send a message over $\delta(C)$, then Alice sends the corresponding message to Bob. Bob then takes over and does the same. $\Pi_C$ terminates when $\Pi$ terminates. It is easy to check that $\Pi_C$ is a correct protocol for $f_C$ and errs with probability at most $\eps$. Further, the total communication for $\Pi_C$ for an input $Y\in\Sigma^{\cK}$ is exactly
\[\sum_{e\in\delta(C)} c_e(Y,\Pi).\]
Thus, by linearity of expectation, the expected cost of $\Pi_C$ under $\mu_C$ is $\sum_{e\in \delta(C)} x_e$. This along with~\eqref{eq:2-party-lb} proves~\eqref{eq:constraint-sat}, as desired.
\end{proof}

The above immediately implies the following corollary:
\begin{cor}
\label{cor:lpo-cc-lb}
Let $p=(f,G,\cK,\Sigma)$ be a problem and $\mu$ be a distribution on $\Sigma^{\cK}$ such that the following holds for every cut $C$ in $G$
\[D_{1/3,\mu_C}(f_C)\ge \alpha\cdot \left(\sum_{i=1}^{\ell} b^i(C)\right),\]
for some value $\alpha>0$
then the following lower bound holds
\[R(p)\ge \alpha\cdot \lpo(G).\]
\end{cor}

Next, we show we can use Corollary~\ref{cor:lpo-cc-lb} to reprove the following lower bound on the $\ED$ function:
\begin{thm}[\cite{CRR14}]
\label{thm:ed-lb}
\[R(\ED,G,K,\{0,1\}^n) \ge \Omega\left(\frac{\sigma_K(G)}{\log{k}}\right).\]
\end{thm}
\begin{proof}
Let $\mu$ be the distribution that picks $k$ random vectors without replacement from $\{0,1\}^n$. It was shown in~\cite{CRR14} that for for every cut $C$ of $G$, we have $D_{1/3,\mu_C}(\ED_C)\ge\Omega\left(\min(|C|,|K\setminus C|\right)$. The claim then follows from Corollary~\ref{cor:lpo-cc-lb} (for $\ell=1$), Theorem~\ref{thm:mfc-lp} and noting that with the constraints above $\lpo(G)$ is the same as $\lpmfc(G,K)$.
\end{proof}

\section{Proof of Theorem~\ref{thm:sd-st}}

\subsection{A collection of multi-way cuts}

We will consider multi-way cuts of a graph $G=(V,E)$. For our purposes a multi-way cut $C$ of $G$ is a partition of $V$ into at least two sets. For notational convenience, we will list all but one set of a multi-way cut $C$: i.e. the ``missing" set will be implicitly defined by the set $V\setminus \cup_{S\in C} S$. Just for concreteness, we will call the sets explicitly mentioned in $C$ is {\em explicit sets} and the missing set to be the {\em implicit set}. (Note that this implies that the size of a multi-cut $|C|$ is the number of explicit sets in $C$.) Also $\cutE{C}$ denotes the set of {\em cut-edges} of $C$: i.e. the set of edges that have one end point in one explicit set of $C$ and the other end point in another set (explicit or implicit) of $C$.

Given two multi-way cuts $C$ and $C'$ of $G$, we say that $C$ is {\em contained} in $C'$ is every explicit set of $C'$ is the union of one or more explicit set of $C$ (and maybe some extra elements from the implicit set of $C$).

We now define a family of collection of multi-way cuts that will be useful in proving our lower bounds.

\begin{defn}
\label{def:multicut-family}
We call a family of collection of multi-way cuts $\cC_1,\dots,\cC_{\ell}$ to be $(\ell,\alpha)$-multicut family for $G$ if the following is true for every $i\in [\ell]$. (For every $i\in [\ell]$, let $\cC_i=\{C_i^{(1)},\dots,C_i^{(m_i)}\}$, where each $C_i^{(j)}$ is a multi-way cut for $G$.)
\begin{enumerate}
\item[(i)] ({\sf Containment property}) For every $1\le j < m_i$, $C_i^{(j)}$ is contained in $C_i^{(j+1)}$.
\item[(ii)] ({\sf Disjointness property}) For every $1\le j_1\neq j_2\le m_i$, $\cutE{C_i^{(j_1)}}$ and $\cutE{C_i^{(j_1)}}$ are disjoint.
\item[(iii)] ({\sf Singleton property}) Call an explicit set $S$ in $C_i^{(j)}$ for any $j\in [m_i]$ to be {\em singleton} if $S$ contains exactly one set from $C_i^{(1)}$. Then $C_i^{(m_i)}$ has at least $\alpha\cdot \abs{C_i^{(1)}}$ singleton explicit sets.
\end{enumerate}
\end{defn}

\subsection{Multicut family to a lower bound}

Next we show how an $(\ell,\alpha)$-multicut family implies a lower bound for certain functions. We begin with the specific class of functions.

Recall that $f:\Sigma^K\to \{0,1\}$ is $h$-{\em \maxhard} if the following holds for any multicut $C$ of $K$. There exists a distribution $\mu_C^f$ such that the expected cost (under $\mu_C^f$) of any protocol that correctly computes $f$ on any star graph where each of the leaves has terminals from an explicit set from $C$ (and the center contains the implicit set of $C$) is $\Omega(|C|\cdot h(|\Sigma|))$.

\begin{lemma}
\label{lem:multicut=>sd-lb}
Let $\cC$ be an $(\ell,\alpha)$-multicut family for $G$ such that every (explicit) set in $C_1^{(1)}$ has at least one terminal from $K$ in it and let $f:\Sigma^K\to \{0,1\}$ be an $h$-\maxhard\ function. Then
\[R(f,G,K,\Sigma)\ge \Omega\left(\frac{\alpha\cdot h(|\Sigma|)}{\ell\cdot \log{k}}\cdot \sum_{i=1}^{\ell} m_i\cdot\abs{C_i^{(1)}}\right).\]
\end{lemma}
\begin{proof}

Fix an $i\in [\ell]$. We will define a hard distribution $\mu_i$ for terminals in $K$ such that the expected cost of communication over all the crossing edges in the multi-way cuts in $\cC_i$ for any correct protocol will be
\begin{equation}
\label{eq:one-level}
\Omega\left( \frac{1}{\log{k}}\cdot\alpha\cdot h(|\Sigma|) \cdot m_i\cdot\abs{C_i^{(1)}}\right).
\end{equation}
Note that by picking the final hard distribution $\mu=\frac{1}{\ell}\sum_{i=1}^{\ell} \mu_i$, will complete the proof.

To complete the proof, we argue~\eqref{eq:one-level}.
Let $s$ be the number of singleton sets in $C_i^{(m_i)}$. Then by the containment property of $\cC$, this implies that there exist explicit sets $T_1,\dots,T_s\in C_1^{(1)}$ such that for every $1<j\le [m_i]$, $C_i^{(j)}$ has $s$ singleton sets that contain $T_1,\dots,T_s$ respectively. We then let $\mu_i$ be $\mu_{\{T_1,\dots,T_s\}}^f$, where we think of $\{T_1,\dots,T_s\}$ as a multicut on $K$.

By the definition of $\mu_{\{T_1,\dots,T_s\}}^f$,
we get that the expected amount of communication on the cut edges $\cutE{C_i^{(j)}}$ (for any $j\in [m_i]$) is $\Omega(sh(|\Sigma|)/\log{k})$.\footnote{The definition implies a lower bound on a star but it is easy to see that any protocol on any connected graph can be simulated on a star graph with only a $O(\log{k})$ blowup in the total communication. In particular, consider the following simulation. When a message needs to be sent from one of the $k$ nodes $u$ to another $v$, the leaf corresponding to $u$ in the $k$-star uses $O(\log{k})$ bits to identify the  leaf corresponding to $v$ to the center so that the center can relay the original message from $u$ to $v$.} Since the cut edge sets are disjoint for any two cuts $C_i^{(j_1)}$ and $C_i^{(j_2)}$, by linearity of expectation, the expected cost over all edges in $\cup_{j=1}^{m_i} \cutE{C_i^{(j)}}$ is $\Omega(sm_ih(|\Sigma|)/\log{k})$. The proof is complete by noting that the Singleton property of $\cC$ implies that $s\ge \alpha\cdot \abs{C_i^{(1)}}$.
\end{proof}

\subsection{Constructing the multicut family}

The main result in this section is to show that we can construct a good multicut family.

\begin{lemma}
\label{lem:good-multicuts}
For any given instance $(G,K)$ there exists an $(\ell=O(\log{k}),\alpha= 1/3)$-multicut family for $G$ such that $C_1^{(1)}=\{\{i\}|i\in K\}$: i.e. all the explicit sets in $C_1^{(1)}$ just contain one terminal from $K$. Further, we have
\[\sum_{i=1}^{\ell} \sum_{j=1}^{m_i} \abs{C_i^{(j)}}\ge \Omega\left(\ST(G,K)\right).\]
\end{lemma}

Note that Lemmas~\ref{lem:multicut=>sd-lb} and~\ref{lem:good-multicuts} prove Theorem~\ref{thm:sd-st}.

In the rest of the section, we prove Lemma~\ref{lem:good-multicuts}. We will in fact first define a collection of multi-way cuts $C_1,\dots,C_t$ for some $t\ge 1$ such that they satisfy the containment and disjointness properties in Definition~\ref{def:multicut-family} for $\ell=1$ (but not necessarily the singleton property). Further, these cuts satisfy the two extra properties needed in Lemma~\ref{lem:good-multicuts}. Finally, we will show how to divide the collection of multicuts into $O(\log{k})$ sub-collections so that the new family is actually an $(O(\log{k}),1/3)$-multicut family for $G$ (without losing the other desired properties).

We start with a notation that will help us define our multi-way cut family. For any non-empty subset $S\subseteq V$, let $\ball_G(S,r)$ denote the set of all vertices in $G$ with a (shortest path) distance of at most $r$ from some node in $S$. More precisely:
\[\ball_G(S,r)=\left\{u\in V| \text{ there exists a } w\in S\text{ such that } d_G(u,w)\le r\right\}.\]

We will define the multicuts $C_1,\dots,C_t$ by defining a partition of $K$ for each $i\in [t]$: let us call the $i$th partition $\cS_i$. 
Given the partition $\cS_i$, the definition of the multicut $C_i$ is simple: there is one explicit set in $C_i$ corresponding to each $S\in \cS_i$. In particular, for every $S\in\cS_i$, we have
\[C_i=\{\ball_G(S,i-1)|S\in \cS_i\},\]
where recall we only state the explicit sets in the multi-way cut $C_i$.

Thus, to complete the descriptions of the multi-way cuts, it is enough to show how to compute $\cS_{i}$.
$\cS_1$ is defined to be the partition of $K$ into the $k$ singleton sets $\{i\}$ (for every $i\in K$). 
To compute $\cS_{i+1}$ from $\cS_i$ we first construct a graph $G'_i$ which has one node for every $S\in \cS_i$. Add an edge $(S,T)$ for $T\neq S\in \cS_i$ in $G'_i$ if 
$\ball_G(S,i)$ intersects $\ball_G(T,i)$. For each connected component in $G'_i$, add the union of all sets from $\cS_i$ in the connected component as one set in $\cS_{i+1}$. Note that it is possible that $\cS_{i+1}=\cS_i$. The last index $t$ is defined as the smallest index such that $|\cS_{t+1}|=1$.

Note that the containment and disjointness properties of the multi-way cuts $C_1,\dots,C_t$ follow from construction. Further, by definition, all the explicit sets in $C_1$ contains exactly one terminal from $K$. Next we argue that

\begin{lemma}
\label{lem:sum-cost=st}
\[\sum_{i=1}^t |C_i| \ge \frac{1}{2}\cdot \ST(G,K).\]
\end{lemma}
\begin{proof}
Let $\bar G$ denote the complete graph on the vertex set $K$, where the edge $(u,v)$ in $\bar G$ has a cost of $d_G(u,v)$. Let $T(\bar G)$ denote an MST of $\bar G$. It is easy to see the cost of $T(\bar G)$ (denoted by $\cost(T(\bar G))$) is at least $\ST(G,K)$. Next we argue that $\sum_{i=1}^t |C_i|$ is at least half of the cost of $T(\bar G)$, which would complete the proof.

Intuitively, the argument about the cost of $T(\bar G)$ is essentially that our algorithm to compute the various $\cS_i$ simulates a run of Bor\.uvka's algorithm~\cite{boruvka} for computing an MST of $\bar G$.

We will now prove the result by induction on $k\stackrel{\mathrm{def}}{=}|K|$. When $k=2$, then it is easy to see that $\sum_{i=1}^t |C_i| \ge \cost(T(\bar G))-1 \ge \cost(T(\bar G))/2$, as desired.\footnote{The inequality holds as long as $\cost(T(\bar G))\ge 2$. If $\cost(T(\bar G))=1$, then note that $\bar G$ is just a unit cost edge with the two end points being the two terminals. Note that in this case $|C_1|=1$ and hence the inequality $\sum_{i=1}^t |C_i| \ge \cost(T(\bar G))/2$ still holds as required.} Let us assume that the claim is true for all $K$ with $|K|=k\ge 2$. 

Next consider the case when $|K|=k+1$. Let $i$ be the smallest index where the graph $G'_{i-1}$ has at most $k$ components (i.e. this is the first $i$ such that at least two singletons sets from $\cS_{i-1}$ are merged when computing $\cS_i$). Let $G'$ denote the graph where we collapse all nodes in $\cC_i$ into ``super-nodes" and let $K'$ denote the corresponding set of terminals in $G'$: i.e. $K'$ is in one to one correspondence with $\cS_i$.
Let $C'_1,\dots,C'_{t'}$ denote the cuts defined if our algorithm ran on $G'$ and $K'$. We claim two properties: (i) $\sum_{i=1}^t |C_i|-\sum_{j=1}^{t'}|C'_j|=(k+1)\cdot(i-1)$ and (ii) the corresponding graph\footnote{In particular, $\bar G'$ is a complete graph on $\cS_i$ and the cost of an edge $(u',v')$ is $d_G(u',v')-2(i-1)$, where $d_G(u',v')$ is the distance between the closest pairs of terminals in $u'$ and $v'$ (recall that $u'$ and $v'$ correspond to disjoint subsets of $K$).} $\bar G'$ has its MST cost (denoted by $\cost(T(\bar G'))$) to be  at least $\cost(T(G))-2k(i-1)$. Note that claims (i) and (ii) complete the inductive step of the proof.\footnote{It can be verified that our construction of $C_i,\dots,C_t$ on $G$ corresponds to running our algorithm on $G'$ with the terminal set $K'$. This implies (e.g. by induction) that $\sum_{j=i}^{t'} |C'_j|\ge \cost(T(\bar G'))/2$. Hence by (i) we have $\sum_{i=1}^t |C_i|\ge (k+1)(i-1)+\cost(T(\bar G'))/2\ge (k+1)(i-1) + \cost(T(G))/2-k(i-1)\ge \cost(T(G))/2+(i-1)\ge \cost(T(G))/2$, where the second inequality follows from (ii).}
 To complete the proof, we argue these two claims.

We begin with claim (i). We first note that the multi-way cut $C'_j$ (for $j\in [t']$) is in one to one correspondence\footnote{This follows by our earlier observation that we can think of the construction of $C_i,\dots,C_t$ as running our algorithm on $G'$ with the terminal set being $K'$.} with $C_{i+j-1}$. In particular, we have $|C'_j|=|C_{i+j-1}|$. The claim then follows by noting that all $\cS_{\ell}=\cS_1$ for $\ell<i$ (and hence $\sum_{\ell=1}^{i-1} |C_{\ell}|=(k+1)(i-1)$).

We finish by arguing claim (ii). The main observation is that $T(\bar G)$ can be obtained by starting with $T(\bar G')$ and then replacing each super node in $T(\bar G')$ by a spanning tree of the corresponding component of $G'_{i-1}$ (recall that each super node in $K'$ is constructed by collapsing a component in $G'_{i-1}$ of size at least two). To complete the claim, we need to track the changes in edge weights. We first note that the cost of edges in $\bar G'$ is smaller than the corresponding edge in $\bar G$ by exactly $2(i-1)$. Second, each edge added back for each super node in $K'$ has cost at most $2(i-1)$. This implies that
\[\cost(T(\bar G))-\cost(\bar G')= \left(|K'|-1\right)\cdot 2(i-1)+ \left(|K|-|K'|\right)\cdot 2(i-1) \le 2k(i-1),\]
as desired.
\end{proof}

Note that now we have shown an $(1,1/k)$-multicut family that satisfies all the other conditions in Lemma~\ref{lem:good-multicuts}. We now present a simple way to convert this into an $(O(\log{k}),1/3)$-multicut family. In particular, we will group $\ell=O(\log{k})$ consecutive chunks of multi-way cuts from $C_1,\dots,C_t$ to obtain our final family $\cC_1,\dots,\cC_{\ell}$. We first show how we compute $\cC_1$. Let $j$ be the largest index in $[t]$ such that $\cS_j$ has at least $k/3$ singleton sets. Then $\cC_1=\{C_1,\dots,C_j\}$. Now note that $|\cS_{j+1}|\le 2k/3$ (because it has at most $k/3$ singleton sets and the rest in the worst-case might form subset of size $2$). We now re-start the process from $C_{j+1}$, where we think of $\cS_{j+1}$ as the set of terminals. If this process stops in $\ell$ steps note that this results in an $(\ell,1/3)$-multicut family. Recall that once we go from $\cC_i$ to constructing $\cC_{i+1}$, the number of terminals decreases by a factor of at least $3/2$. This in turn implies that $\ell=O(\log{k})$, as desired.


\end{document}